\documentclass[journal,onecolumn]{IEEEtran}
\usepackage{comment}
\usepackage[utf8]{inputenc} 
\usepackage[T1]{fontenc}
\usepackage{url}
\usepackage{ifthen}
\usepackage{cite}
\usepackage[cmex10]{amsmath} 
\usepackage{amssymb}
\usepackage{amsmath}
\usepackage{amsthm}
\usepackage{etoolbox}
\usepackage{bbm}
\usepackage{amsfonts}
\usepackage{graphicx}
\usepackage[dvipsnames]{xcolor}
\usepackage{algpseudocode}
\usepackage{algorithm}
\usepackage{multirow}
\algnewcommand\algorithmicforeach{\textbf{for each}}
\algdef{S}[FOR]{ForEach}[1]{\algorithmicforeach\ #1\ \algorithmicdo}

\usepackage{amsmath}

\DeclareMathOperator*{\argmin}{arg\,min}

\newtheorem{proposition}{Proposition}[section]

\newtheorem{theorem}{Theorem}[section]

\newtheorem{remark}{Remark}[section]

\theoremstyle{definition}

\newtheorem{definition}{Definition}[section]

\newtheorem{corollary}{Corollary}[proposition]

\newtheorem{lemma}[theorem]{Lemma}

\makeatletter
\newcommand{\vast}{\bBigg@{4}}
\newcommand{\Vast}{\bBigg@{5}}
\makeatother

\ifCLASSINFOpdf
\else
\fi
\begin{document}
%
\title{Guessing Cost: Bounds and Applications to Data Repair in Distributed Storage}
%
%
%

\author{Suayb~S.~Arslan,~\IEEEmembership{Senior Member,~IEEE,}
        and~Elif~Haytaoglu,~\IEEEmembership{Member,~IEEE}
\thanks{{S. S. Arslan is the corresponding author and currently affiliated with the Department
of Brain and Cognitive Sciences, Massachusetts Institute of Technology, Cambridge, MA, USA. e-mail: sarslan@mit.edu.}}
\thanks{E. Haytaoglu is  with the Department
of Computer Engineering, Pamukkale University, Denizli, Turkey.}
\thanks{This work is partially presented in {two different} IEEE International
Symposium on Information Theory (ISIT) conferences which were held in
2020 (Los Angeles, USA) { and in 2022 (Espoo, Finland)}.}
\thanks{This work is supported by
TUBITAK under grant number 119E235.}}

%
%

\markboth{Submitted to IEEE Transactions on Information Theory ,~Vol.~NA, No.~NA, March~2022}%
{Shell \MakeLowercase{\textit{et al.}}: Bare Demo of IEEEtran.cls for IEEE Journals}
%



\maketitle

\begin{abstract}
The guesswork refers to the distribution of the minimum number of trials needed to guess a realization of a random variable accurately. In this study, a non-trivial generalization of the guesswork called \textit{guessing cost} (also referred to as cost of guessing) is introduced, and an optimal strategy for finding the $\rho$-th moment of guessing cost is provided for a random variable {defined} on a finite set whereby each choice is associated with a positive finite cost value {(unit cost corresponds to the original guesswork)}. Moreover, we drive asymptotically tight upper and lower bounds on the {logarithm of guessing cost moments.}  Similar to previous studies on the guesswork, established bounds on the moments of guessing cost quantify the accumulated cost of guesses required for correctly identifying the unknown choice and are expressed in terms of Rényi's entropy. Moreover, new random variables are introduced to establish connections between the guessing cost and the guesswork, {leading to induced strategies}. Establishing this implicit connection {helped us obtain} improved bounds {for} the non-asymptotic region. As a consequence, we establish the \textit{guessing cost exponent} in terms of Rényi \textit{entropy rate} on the moments of the guessing cost using {the} optimal strategy by considering a sequence of independent random variables with different cost distributions. Finally, with slight modifications to the original problem, these results are shown to be  {applicable} for bounding the overall repair bandwidth for distributed data storage systems {backed up by base stations and protected by bipartite graph codes.} 
\end{abstract}


%
\IEEEpeerreviewmaketitle

\section{Introduction}
%
%
%
%
\IEEEPARstart{T}{he} typical guessing framework involves finding the
value of a realization of a random variable $X$ from a finite or countably infinite set  $\mathcal{X}$ by asking
a series of questions ``Is $X$ equal to $x \in \mathcal{X}$?'' until the answer becomes
“Yes”. {What makes guessing framework challenging} is that each answer typically {affects} the following questions and associated answers {in which} the number of questions {is not necessarily} fixed {a priori, whereas} the questions {are determined based on a fixed strategy} before the decision about the guess is finalized. 

{Given the distribution of $X$, denoted by $P_X(x)$,} the ultimate objective of  guessing {framework} is to find the distribution of the number of questions (guesses) {before identifying the right answer}. In an attempt to optimize the order of these questions, an optimal guessing strategy i.e., a bijective function from $\mathcal{X}$ to a finite or countably infinite set $[|\mathcal{X}|] \triangleq \{1,\dots,|\mathcal{X}|\}$ is adapted to typically minimize the average number of guesses, also known as the average \textit{guessing number}. In \cite{pliam99}, this problem is named as \textit{guesswork} and lower and upper bounds are investigated on the guessing number in terms of Shannon's entropy by Massey \cite{massey1} and later on by McElice and Yu \cite{mceliceyu}. A sequence of independent and identically  distributed random variables $X_1,\dots,X_n$ are considered for practical applications and asymptotically tight bounds are derived on the moments of the expected number of guesses for {the} guesswork \cite{arikan1996}. This study has related the asymptotic exponent of the best achievable guessing
moment to the Rényi's entropy. {Later, bounds on the moments of optimal guessing are improved in \cite{Bozdas1997} and subsequently in \cite{sason2018}. Particularly, the relationship between Rényi's entropy and average guessing number is interesting and useful in different engineering contexts.}  In fact, Rényi's entropy was a frequently used information measure in different contexts such as source coding to be able to generalize coding theorems in the past \cite{cambell56}. Such findings on the derived bounds are successfully applied to various recent applications of data compression \cite{kuzuoka2019}, channel coding \cite{duffy2019}, networking and data storage security \cite{bracher2019} through tweaking the original guesswork problem so that it fits within the requirements of the  application at hand.

In many practical scenarios, making a guess about the state of a system (in a physical realm) or the unknown value of a random variable ({both} in presence {or absence} of side information \cite{arikan1996} {or compressed side information \cite{Graczyk2022}} {might} lead to a certain amount of cost. {In our work,} using the same set-up, unlike the guesswork, the random variable will be associated with the cost set $\mathcal{C} = \{c_1,\dots,c_{|\mathcal{C}|}\}$ and the additive term in the computation of average cost for $x\in \mathcal{X}$ would be $\sum_x c_xP_X(x)$. Hypothesizing this cost measure to typically represent the potential risks and consequences that may arise from making an incorrect guess, unity costs will correspond to the guesswork as the baseline. With this generalization, the impact of making an error could be high enough to justify taking the time and effort to gather more information or perform additional measurements to reduce the uncertainty.  {Consequently}, making a choice among multiple possibilities may lead to different types and amounts of costs overall where we would refer it as the \textit{cost of guessing} or simply \textit{guessing cost} throughout. In {general}, these costs may dynamically be changing after making subsequent guesses about a series of random variables $\{X_i\}_{i=1}^n$ {not necessarily independent}. Independent and identically distributed random variables within the context of guesswork is thoroughly studied in the literature and some extensions to {ergodic} Markovian dependencies are also considered \cite{malone2004}. {More complex dependencies such as the one formed by shift spaces are considered for a sequence of random variables in \cite{pfister2004}}. To our best of knowledge, the cost of guessing is only mentioned recently in \cite{kuzuoka2019} in a limited context whereby the guesser is allowed to stop guessing and declare an error and only then a fixed amount of cost is applied, otherwise the mechanism is identical to the guesswork. {In addition,} the {definition of} ``cost'' is expanded in the context of guessing in \cite{arslanisit2020} {to cover each choice to have an individual numerical cost value} {and a few improved bounds are provided later in \cite{arslanisit2022}}.

\subsection{Background and {Past Applications}}

{There are numerous applications involving guesswork and its various uses, with many of them being related to cryptography and data correction. For instance,} in the field of security known as the {(public)} keyword guessing \cite{noroozi2019} around the cryptographic notion called \textit{searchable encryption} \cite{Song2000}. On the other hand, guessing is recognized to be a useful analysis tool for data detection and error correction coding as well. For instance, it is shown that the cut-off rate of sequential decoding can easily be characterized if guessing theory is applied to the general idea of decoding of a tree code \cite{arikan1996}. {The application of guessing to coding theory dates back to Ulam's problem \cite{Ulam1988} where one is allowed to lie in their responses \cite{ArikanGwLies}.} More recently, capacity-achieving maximum likelihood decoding algorithms are developed in a data communication context based on guessing \cite{duffy2019}. Later, these studies evolved to develop universal decoders especially for low-latency communication scenarios \cite{duffy2022}. Moreover, thanks to the optimal strategy  that lists most likely noise sequences and low implementation complexity, guessing framework is demonstrated to be a viable decoding option for the control channel of 5G networks \cite{solomon2020}. Therefore, {we believe that} extending the idea of guessing by associating a cost with each choice would be quite powerful and will find plenty of interesting applications in communications engineering of future generation standards. Distributed systems constitute yet another application area in which cost of data communication depends on the link loads, node availabilities and current traffic at the time of data communication etc. Moreover base stations (BS) could also be used to help with the network data reconstruction processes at the expense of increased costs \cite{haytaoglu2022}. {BS can help reduce the time needed to reconstruct data, as well as reduce the average cost of the process.} Such costs can be expressed in terms of latency, bandwidth used to transfer information or computation complexity  depending on the context. 

{
Recently, there have been a multitude of diverse research efforts undertaken to investigate the general concept of guessing across various domains. For instance, in \cite{Brocher2015}, security attacks on distributed storage systems in which an attacker can use one hint on the  sensitive data is analyzed. In another study \cite{Christiansen2015}, bounds for a specific setting in which simultaneous guesses can be made is investigated. Furthermore, the study presented in \cite{Huleihel2017} delves into the subject of restricted-memory guessing, wherein the individual making the guesses is limited by their ability to recall their previous attempts. In  a more recent study, \cite{Kumar2022} focuses on developing general framework on well-known problems related to guesswork such as, source coding, task partitioning, etc. Among these studies, both \cite{Christiansen2015} and \cite{Huleihel2017} can be extended with the assignment of ``costs" for guessing random variables. As for the case in \cite{Huleihel2017},  the representation size of each variable can be associated with a form of cost value since guessing a particular value can incur different usage patterns in memory, whereas the number of simultaneous guesses, i.e. the number of attacking computers, can cause additional cost in the system. The study in \cite{Christiansen2015} can be further extended by taking into account the same specific parameter.
}

\begin{figure}[t!]
    \centering
    \includegraphics[width=0.75\textwidth]{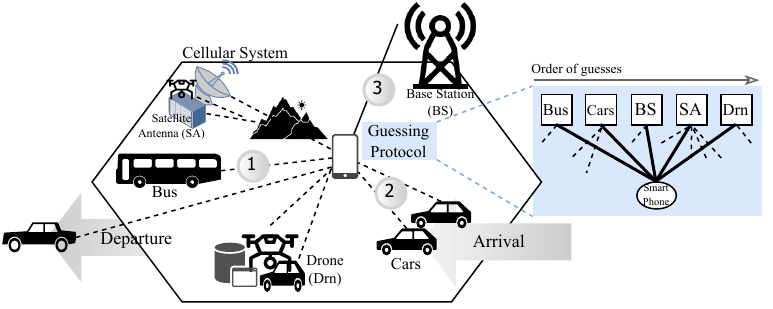}
    \caption{{A motivating example for data regeneration scenario using a protocol based on the guessing cost. The cellular phone (a node) first attempts the data-carrying bus ({\large \textcircled{\small 1}}), then two cars ({\large \textcircled{\small 2}}), respectively, since downloading data from two different cars is more expensive than from a bus. Upon unavailability or download failure, the node finally retrieves the data from the base station ({\large \textcircled{\small 3}}). Each connectivity has its own cost measure and is a function of the code used to protect data. For instance, at time of the regeneration, the drone group (due to mobility) and satellite antenna site (due to obstructions) connections were too costly to be used.}}
    \label{fig:motexample}
\end{figure}

\subsection{{Motivating Example and Contribution}}

{In various applications of dynamic cellular networks, data is partitioned and disseminated among multiple nodes of the system, with nodes joining and leaving the cell at unpredictable times. An example scenario is provided in Fig. \ref{fig:motexample}. Because of the frequently changing status of the nodes in the cell, it is difficult to be informed of data location in real time. The data on the departed nodes may be lost indefinitely, necessitating data regeneration, while the system only has minimal knowledge about the cached data whereabouts. Data may be downloaded from either local nodes or the base station, with the cost of each option depending on the physical distance of the nodes, potential obstructions for the line of sight, available bandwidth, or even the popularity of the file being regenerated/cached.}

{As can be seen in Fig. \ref{fig:motexample}, a guessing protocol is required to run in a smartphone to regenerate the needed data piece. As can be seen, the protocol makes two attempts based on a predefined strategy to locate the data. Upon unavailability, the third attempt has been to download it from the base station. Another intriguing use case could be searchable encryption \cite{Song2000}. One can typically spend more time searching for a cipher keyword in an encrypted document depending on its size or number of defined keywords, leading to varying processing requirements. We finally notice that all such cost considerations can be integrated into our generalized guessing cost {framework}.}

Motivated from such {examples}, in this study, we {introduce} the guessing cost, and derive optimal guessing strategies {(the ones that minimize {the various statistics about the cost}  as well as asymptotically tight bounds by using a quantity related to the Rényi's entropy {for the expected values of real powers a.k.a. moments of the guessing cost}.
Note that numerical calculation of moments of guessing cost might be computationally feasible for small $|\mathcal{X}|$, however {when} we consider a series of random variables $\{X_i\}_{i=1}^n$ {each defined on the same set $\mathcal{X}$,} finding the optimal guessing strategy for minimum average {or} the moments of {the guessing} cost would be computationally intractable (exponential in $n$), {motivating us for finding tight bounds}. We also have shown that {bounds on the} guessing cost of a sequence of independent random variables (not necessarily identically distributed) can be expressed in terms of \textit{Rényi entropy rate} {(Theorem \ref{thmexponent})}, {which is defined for order-$\alpha$ ($\alpha \in {\mathbb{R}^+}$) as
{
\begin{align}
    \mathcal{R}_\alpha (\{X_i\}) \triangleq 
    \begin{cases}
   \lim_{n\rightarrow \infty} \frac{H(\{X_i\})}{n}, & \text{ if } \alpha = 1,\\
     \lim_{n\rightarrow \infty} \frac{H
_\alpha(\{X_i\})}{n}, & \text{Otherwise}
    \end{cases}
\end{align}}
{as long as} the limits exist, where $H
_\alpha(\{X_i\})$ is the joint Rényi entropy \cite{renyi1960} {and $H(\{X_i\})$ is the Shannon entropy} of the sequence $\{X_i\}_{i=1}^n$}. On the other hand, the computation of the moments of guessing cost for independent {sequence of} random variables $\{X_i\}_{i=1}^n$ is observed to be linear in $n$. Our results are asymptotically tight i.e., as $n \rightarrow \infty$ we characterize the exponential growth rate of the moments of guessing cost. Several improved bounds are conjectured for the non-asymptotic region based on an established connection with the guesswork. {Moreover,} we realized that our findings for guessing cost can be easily applied to {an example} distributed data storage scenario, {as depicted in Fig. \ref{fig:motexample}}, where nodes are repaired using a master-based {regeneration} protocol and graph-based codes such as low density parity check (LDPC) codes \cite{LDPC1963} in the event of node failures or unexpected node departures from {a base station's} coverage \cite{haytaoglu2022}. 

\subsection{Organization}

The rest of the paper is organized as follows. In Section II, the problem is formally stated and necessary and sufficient conditions are laid out for an optimal guessing strategy that minimizes the moments of guessing cost. In addition, {distinct guesses} for costs are discussed along with an algorithm that describes an optimal guessing strategy for non-configurable costs. In Section III, tight upper and lower bounds are provided {for  guessing cost moments of a random variable and the logarithm of the guessing cost moments of a series of random variables}. While deriving the upper bound, new random variables are introduced and the connection with the guesswork is established. This connection helped us observe that the previous findings may be utilized to find/characterize tighter bounds in the non-asymptotic regime for the guessing cost. In Section IV, {an example} distributed storage scenario is considered where a long--blocklength LDPC code is utilized along with a {guessing} protocol which uses a master node/base station to help with the data regeneration process. It is shown that the data repair problem of an LDPC code can be considered within the context of guessing cost. Several numerical results are provided  before we conclude our paper in Section V.

\subsection{Notation}

In our work, $X$ denotes the random variable whose values are selected from the set $\mathcal{X}$ according to a probability distribution i.e., $X \sim P_X(x)$ where $\sim$ denotes ``distributed''. We use {$| \ . \ |$} to denote the cardinality of a set or absolute value depending on the context and $[M]$ to mean the {index} set $\{1,2,\dots, M\}$. $\mathbb{E}[.]$ is the expectation operator. We denote the order-$\alpha$  Rényi's entropy by $H_\alpha(X)$ and Rényi's rate by $\mathcal{R}_\alpha(X)$ for a given random variable $X$. Also, $(z)^+ \triangleq \max\{z, 0\}$ for $z \in \mathbb{R}$. For a given length-$n$ vector $\mathbf{x}=(x_1,\dots,x_n)$, the $a$--norm is defined to be $||\mathbf{x}||_a = \left(\sum_{i=1}^n |x_i|^a \right)^{1/a}$ for any positive real $a \geq 1$. Also for integers $c$ and $l\leq n$, we define $\mathbf{x}^{(l)} \triangleq (x_1,x_2,\dots,x_l)$ and $\mathbf{x}^{(l)}+c \triangleq (x_1+c,x_2+c,\dots,x_l+c)$.

\section{Problem Statement and Guessing Strategy}

Let $C_\mathcal{G}(x)$ denote the guessing cost required by a particular guessing strategy $\mathcal{G}: \mathcal{X} \rightarrow [M]$ when $X = x$ (the realization of the random variable $X$ is $x$). If the cost of making {the} guess $X = x$ is independent of other guesses, {each having unit cost}, then this problem would be the same as the characterization of the average number guesses (\textit{expected guessing number}) and is identical to Massey's original guessing problem introduced earlier \cite{massey1}. 

Let us assume that the random variable $X$ can take on values from a finite set $\mathcal{X} = \{x_1,\dots,x_M\}$ \ according to a distribution $P_X(x)$ with  {the associated set of costs $\mathcal{C} = \{c_{i}\in \mathbb{R}^+, c_i \geq 1 | 1 \leq i \leq M\}$ and { $M \in \mathcal{Z^+} $}.}
Without loss of generality, sets are assumed to have cardinalities $|\mathcal{X}|=|\mathcal{C}|=M$ in which using a particular guessing strategy $\mathcal{G}$, the probability that a randomly selected element of $\mathcal{X}$ can be found in the $i$-th guess is $p_i = P_X({\mathcal{G}^{-1}(i)})$ with the associated cost {$c_i = c_{\mathcal{G}^{-1}(i)}$}, independently of already made guesses. Then, the average guessing cost using the strategy $\mathcal{G}$ can be expressed as follows
\begin{align}
\mathbb{E}[C_\mathcal{G}(X)] = \sum_{i=1}^M \sum_{j=1}^{i} c_j p_i = \sum_{i=1}^M f_i p_i = \sum_{i=1}^M c_i \left(1 - g_{i-1} \right)
\end{align}
where $f_i=\sum_{j=1}^{i} c_j$ and $g_{i} = \sum_{j=1}^{i}p_j$ are cumulative cost and probability distributions. The minimization of this value is a function of both guessing strategy $\mathcal{G}$  and the probability distribution of $X$.


{In the context of guessing, One of the fundamental questions arises: when provided with the probability distribution $P_X(x)$ and $\rho>0$, what strategy, denoted as $\mathcal{G}^*$, can minimize $\mathbb{E}[C_\mathcal{G}(X)^{\rho}]$, i.e.,
\begin{eqnarray}\label{eq:optimalStrategy}
\mathcal{G}^* \triangleq \argmin_{\mathcal{G}} \mathbb{E}[C_\mathcal{G}(X)^\rho].
\end{eqnarray}}
{For $c_i = 1, \forall i \in [M], \rho=1$}, the optimal strategy is studied and well known i.e., guess the
possible values of $X$ in the order of non-increasing probabilities \cite{massey1}. In other words, without loss of generality, we can assume  $\mathbf{p}^{(M)}$ with $p_1 \geq p_2 \geq \dots \geq p_{M}$ being the probabilities of choosing values from $[M]$ and $\mathcal{G}(X = x_i) = i$. Then with this choice, the quantity $\sum_i i p_i$ would be minimized. However, the same conclusion cloud not be easily drawn for an arbitrary vector of costs $\mathbf{c}^{(M)}$ {$\in \{ {\mathbb{R}^+ \setminus [0,1)\}}^{M}$}. 
Let us consider two possible cost selection scenarios based on the timing with respect to when exactly the guesses are made, that will have different implications.  


In the {first} scenario, although the cost {values are}  given i.e., $\mathcal{C}$, the assignments are not made {a priori} i.e., costs can be associated with each choice as it minimizes the average/moments of guessing cost in the beginning of {or during} the guessing process. {To illustrate this scenario, let us assume that we are in the situation of transporting a water tank to a fireplace. The collection of M tanks at our disposal, despite having the same capacity, is made of different materials, incurring different expenses. Furthermore, we have $M$ different vehicles, each with a different chance of successfully reaching the firing zone. In this situation, our primary purpose is to transport a single tank to the given site. Such a work gives us the freedom to attach any tank to any vehicle, with preset prices associated to each choice but costs that may vary based on the precise tank and vehicle coupling chosen.} {Given the configurable costs}, the best strategy  for $\rho=1$ is to guess the
possible values of $X$ in the order of non-increasing probabilities and associate the more probable choice with the smallest cost value. In other words, for any assignment (a permutation of $\mathbf{c}^{(M)}$) $\Tilde{\mathbf{c}}^{(M)} = (\tilde{c}_1,\tilde{c}_2,\dots,\tilde{c}_M)$ with $\tilde{c}_1 \leq \tilde{c}_2 \leq \dots \leq \tilde{c}_M$ and $\tilde{c}_i \in \mathbf{c}^{(M)}$, it is easy to see that for any $\rho > 0$, we have
\begin{eqnarray}
\sum_{i=1}^M \left( \sum_{j=1}^{i} c_j \right)^\rho p_i \geq \sum_{i=1}^M \left( \sum_{j=1}^{i} \tilde{c}_j \right)^\rho p_i.
\end{eqnarray}

In case the cumulative costs are given by the moments of the guessing number i.e., $f_{i}=i^\rho$  for any $\rho \geq 1$, then it is easy to see that $c_i = i^\rho - (i-1)^\rho$ which implies that $c_1 \leq c_2 \leq \dots \leq c_M$ is satisfied. Thus, the {optimal} strategy would again be to guess the possible values of $X$ in the order of non-increasing probabilities as argued in \cite{arikan1996}. 



{The second scenario, in which costs associated with each choice are externally determined, involves finite costs and choices predetermined before the guessing process begins. This typical scenario is examined throughout the article.}  In this case, the best strategy for {$\rho=1$} would not necessarily be guessing the 
possible values of $X$ in the order of non-increasing probabilities. 

\textbf{{Example:}} {Suppose that} there are three choices $1,2,3$ ($M=3$) and $\rho=1$ with $\{1,p_1=0.5,c_1=20 \}$, $ \{2,p_2=0.4,c_2=2 \}$ and $\{3,p_3=0.1,c_3=1\}$. In that case the guessing order expressed as the {index} set $\{2,3,1\}$ would be preferable  over the set $\{1,2,3\}$ with the average costs being 12.6 and 21.1, respectively. Note that the latter choice, which is based on the order of non-increasing probabilities, is clearly not optimal.  

We next provide the following proposition {establishing} a necessary condition for the optimal guessing strategy $\mathcal{G}^*$ for non-configurable costs.  

\begin{proposition} \label{prop0}
For a given $\rho > 0$ and an optimal guessing strategy, namely $\mathcal{G}^*$ {(or $\mathcal{G}^*_\rho$\footnote{{The subscript indicates the dependency of the optimal strategy on the choice of $\rho$. However, we omit this notation unless it is absolutely necessary to simplify the notation.}})} for the $\rho$-th moment of guessing cost, we have the following necessary condition for all $i, j \in [M]$ satisfying $i \leq j$,
\begin{align}
    \left[ ||\mathbf{c}^{(i)}||_{1}^\rho - ||\mathbf{c}^{(j)}||_{1}^\rho \right] p_i + \left[ ||\mathbf{c}^{(j)}||_{1}^\rho - (||\mathbf{c}^{(i)}||_{1}-c_i+c_j)^\rho \right] p_j \leq  \sum_{l=i+1}^{j-1}  \left[  (||\mathbf{c}^{(l)}||_{1}-c_i+c_j)^\rho - ||\mathbf{c}^{(l)}||_{1}^\rho \right] p_l.  \label{necessarycon}
\end{align}
\end{proposition}

\begin{corollary} \label{cor211}
Furthermore, if $\rho \geq 1$, then for any $i \leq j$ the condition \eqref{necessarycon} can be simplified to 
\begin{align}
\left[||\mathbf{c}^{(i+1)}||_{1}^\rho - (||\mathbf{c}^{(i+1)}||_{1}-c_i)^\rho\right]p_j \leq \left[||\mathbf{c}^{(j)}||_{1}^\rho - ||\mathbf{c}^{(j-1)}||_{1}^\rho\right] p_i
\end{align}
\end{corollary}

\begin{proof}
The proof of Proposition \ref{prop0} as well as Corollary \ref{cor211} can be found in Appendix A. 
\end{proof}

\begin{remark} \label{prop1}
It is clear from  Corollary \ref{cor211} that for a given optimal guessing strategy for the mean guessing cost (i.e., $\rho=1$), we must have $c_{i}p_{j} \leq c_{j}p_{i}$ for all $i, j \in [M]$ satisfying $i \leq j$. 
\end{remark}
 
\begin{remark}
We also note that there may be more than one optimal guessing strategy that would satisfy the conditions given above. Of these, one or more specific selections will result in the minimum guessing cost. The solution is unique only if the relation in the necessary condition \eqref{necessarycon} is a strict inequality. 
\end{remark}
\begin{remark}
{Note that in our setting due to the freedom of choosing cost values arbitrarily, the choice of $\rho$ may change the stochastic nature of guessing strategies i.e., the strategies for different $\rho$'s satisfying the condition in Proposition 2.1 are not necessarily the same and hence making a stochastic dominance argument across the optimal guessing strategies for distinct moments would not be possible.} 
\end{remark}

In observation of Proposition \ref{prop0}, let us provide an algorithmic solution to finding optimal {strategy for the minimum} guessing cost. We notice that if the order based on the inequality in the proposition is executed using a Bubble-sort\footnote{Bubble-sort is a sorting algorithm that works by repeatedly swapping the adjacent elements in a given list based on a condition.} style for the given strategy, the convergence to an optimal guessing would be guaranteed.  {Here, using Algorithm 1, we can find the optimal solution with the best and the worst time complexities, which is mainly dominated by the sorting processes, with  $\Omega(M)$ and $\Theta(M^2)$,  respectively.} An example naive algorithm that finds an optimal guessing strategy for minimizing the $\rho$-th moment of accumulated cost is provided in \textbf{Algorithm \ref{alg1}} where \texttt{swap}(.,.) function swaps the entries of a given array in the argument. Alternatively, Merge-sort \cite{mergeSort} or Heap-sort \cite{Williams1964algorithm} \ could also be applied, which would result in $\Theta(M \log M)$  worst and the average case complexities respectively. In the next section, we focus on the moments of the guessing cost whereby the average cost would be a special case. Furthermore, lower and upper bounds are derived in terms of a popular information theoretic measure, namely Rényi's entropy.

\begin{algorithm}[t!] \label{alg1}
\begin{algorithmic}[1]
\State \textbf{function}  \textbf{OptimalCostGuess}($\mathbf{p}, \mathbf{c}, \rho$)
\State $M \gets |\mathbf{p}|$ 
\State $\mathcal{I} \gets \lbrace 1,2,3...,M\rbrace$ \Comment{Selection Order}
\State $swapped \gets true$
\While{$swapped$}
    \State $swapped \gets false$
    \For{$j=1:M-1$}
    \If{$\left[||\mathbf{c}^{(j+1)}||_{1}^\rho - (||\mathbf{c}^{(j+1)}||_{1}-c_j)^\rho\right] p_{j+1} >  \left[||\mathbf{c}^{(j+1)}||_{1}^\rho - ||\mathbf{c}^{(j)}||_{1}^\rho\right]p_j$}
    \Comment{If the condition does not hold}
    \State \texttt{swap}($c_{j},c_{j+1}$)
    \State \texttt{swap}($p_{j},p_{j+1}$)
    \State \texttt{swap}($\mathcal{I}_j,\mathcal{I}_{j+1})$
    \State $swapped \gets true$
    \EndIf
   \EndFor
\EndWhile
\State return $\mathcal{I}$ 
\end{algorithmic}
\caption{} \label{alg1}
\end{algorithm}

\section{Bounds on Moments of the Guessing Cost}

Throughout this section and the following sections, we shall assume static costs determined a priori and focus on moments of guessing as the average guessing cost would be a special case. {Let us begin by defining} two auxilary random variables $Y$ and $Z$ based on previously defined random variable $X$ and its associated set of costs {$\mathcal{C} = \{c_{x_1},c_{x_2},\dots,c_{x_M}\}$}. 

\begin{definition} \label{def31}
Let us define the random variable $Y$  that takes on values from a finite set {$\mathcal{Y}=\{y_1,y_2,\dots,y_{\sum_{x=1}^{M} \lceil c_x \rceil }\}$}, and {the probability distribution} of $Y$ is defined as {$P_Y(y_t) = P_X(x_i)/ \lceil c_{x_i} \rceil$ for all positive reals  { $c_{x_i}\geq 1$, $i\in [M]$}  and {$t \in [|\mathcal{Y}|$]}}satisfying  {
\begin{eqnarray}
\sum_{x \in \mathcal{X}^{i-1}}\lceil  c_{x} \rceil < t \leq \sum_{x \in \mathcal{X}^i} \lceil c_{x} \rceil \label{eqnyceiling}
\end{eqnarray}
where $\mathcal{X}^i = \{x_1,\dots,x_i\}$ with $i=0$ corresponding to empty set} and the random variable $Z$  to take on values from a finite set {$\mathcal{Z}=\{z_1,z_2,\dots,z_{\sum_{x=1}^{M}\lfloor c_x \rfloor }\}$} with probabilities 
{$P_Z(z_t) = P_X(x_i)/ \lfloor c_{x_i} \rfloor$ for all $c_{x_i} > 1$, $i\in [M]$ and $t\in[|\mathcal{Z}|]$} satisfying  {
\begin{eqnarray}
\sum_{x \in \mathcal{X}^{i-1}} \lfloor c_{x} \rfloor < t \leq \sum_{x \in \mathcal{X}^i} \lfloor c_{x} \rfloor.
\end{eqnarray}}
\end{definition}{
In Fig. \ref{fig:YZexample}, An example cost set $\mathcal{C}$ (to the right) along with a uniform $P_X(x)$ is assumed  where the corresponding distributions $P_Y(y)$ and $P_Z(z)$ are illustrated (to the left). The same plot also shows the threshold points (red for $Y$, blue for $Z$) where different realizations may have a different probability of occurring. Let us continue with the definition of induced guessing strategy. {Note that the induced strategy is defined for integer costs to make it applicable to both random variables $Y$ and $Z$ at the same time since they are defined based on either ceiling or floor of real costs.}
}

\begin{figure}[t!]
    \centering
\includegraphics[width=1.01\textwidth]{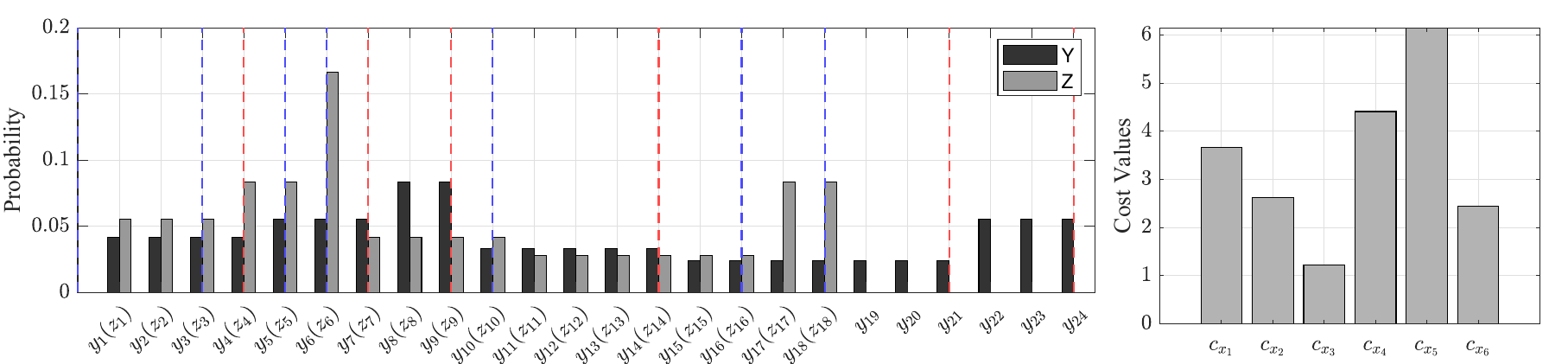}
    \caption{{An example cost set $\mathcal{C} = \{c_{x_i}\}_{i=1}^{6}$ is shown (right plot) and used with a uniform distribution (i.e., {$P_X(x) = 1/6$}). Based on the definitions 3.1 and 3.2, plot on the left demonstrates the calculated distributions $P_Y(y)$ and $P_Z(z)$. Also shown are the index thresholds (cumulative sum of  (floor/ceiling of) costs) for $y\in\mathcal{Y}$ (red) and $z\in\mathcal{Z}$ (blue) across which the assigned probabilities may change. Note that the support for random variables $Y$ and $Z$ are different due to floor/ceiling operations. }}
    \label{fig:YZexample}
\end{figure}

{Any guessing strategy function $\mathcal{G}$ defined on $\mathcal{X}$ can be transformed to one specific guessing strategy using the induced guessing strategy described below, for  $Y$ and $Z$, named as $\mathcal{H}(Y)$ and $\mathcal{F}(Z)$, respectively. }

\begin{definition} \label{inducedG} ({\textit{Induced Guessing Strategy}})
Let us consider a {new} random variable $\overline{X}$ to take on values from a finite set  {{$\overline{\mathcal{X}}=\{\overline{x}_1,\overline{x}_2,\dots,\overline{x}_{\sum_{x=1}^{M}  c_x }\}$}} with probabilities $P_{\overline{X}}(\overline{x})=P_X(x)/ c_x $ {for $\overline{x}\in \overline{\mathcal{X}}$, $x \in \mathcal{X}$ and $c_x \in \mathbb{Z}^+$}. 
Let us further assume  that the guessing strategy $\mathcal{G}$ is used to guess the values of $X$. For a given index {$i \in {[\sum_x  c_x  ]}$} there exists a positive integer $k^{(i)}\leq M$ satisfying
{
\begin{align}
    \sum_{x \in \mathcal{X}^{k^{(i)}-1}}  c_x  < i \leq  \sum_{x \in \mathcal{X}^{k^{(i)}}}  c_x . \label{ineq9}
\end{align}
}
The induced guessing strategy $\overline{\mathcal{G}}$ for guessing the values of $\overline{X}$ is defined to be 
{
\begin{align} 
\overline{\mathcal{G}}(\overline{X} = \overline{x}_i) \triangleq \sum_{x: \mathcal{G}(x) < \mathcal{G}\left(x_{k^{(i)}}\right)}  c_x - \sum_{x \in \mathcal{X}^{k^{(i)}-1}}  c_{x}  + i. \label{eq10}
\end{align}
}
\end{definition}
{
\begin{proposition} \label{prop00}
The induced guessing strategy, namely $\overline{\mathcal{G}}$, is a valid strategy (bijection).
\end{proposition}
\begin{proof}
The proof of Proposition \ref{prop00} can be found in Appendix B. 
\end{proof}}

Using the Definition \ref{inducedG}, {since $\lfloor 
 c_x \rfloor$ and $\lceil c_x \rceil$ are integers,} we can define $\mathcal{F}(Z)$ and $\mathcal{H}(Y)$ to be the \textit{induced guessing strategies} for random variables $Z$ and $Y$, respectively. 
 
 {\textbf{Example:}} {Let us consider the example in Fig. \ref{fig:YZexample}. {Note that} the optimal strategy ($\mathcal{G}^*(X)$) is to {perform selections} in order of non-decreasing costs due to uniform $P_X(X) = 1/6$. {In other words, $\mathcal{G}^*(X=x_1) = 3$, $\mathcal{G}^*(X=x_2) = 2$, $\mathcal{G}^*(X=x_3) = 1$, $\mathcal{G}^*(X=x_4) = 4$, $\mathcal{G}^*(X=x_5) = 6$, $\mathcal{G}^*(X=x_6) = 5$.} For, say $i=15$, it is not hard to verify $k^{(i)}=5$ using inequalities \eqref{eqnyceiling} {since $\sum_{x \in \mathcal{X}^{k^{(i)}-1}}  c_x = 14 < i \leq 21 = \sum_{x \in \mathcal{X}^{k^{(i)}}}$}. Now, using \eqref{eq10}, we compute  
 \begin{align}
      \mathcal{H}^*(Y = y_{15}) &= \sum_{x: \mathcal{G}^*(x) < \mathcal{G}^*\left(x_{5}\right)} {\lceil} c_x {\rceil} - \sum_{x \in \mathcal{X}^{4}} {\lceil} c_{x} {\rceil} + 15 = \underbrace{\sum_{x: \mathcal{G}^*(x) < 6} {\lceil} c_x  {\rceil}}_\text{{$=17$}} - \underbrace{\sum_{i=1}^4 {\lceil} c_{x_i} {\rceil}}_\text{{$=14$}} + 15 = 18.
 \end{align}
 which is in line with the optimal guessing strategy ($\mathcal{H}^*$) induced for $Y$ using the arguments of the guesswork  i.e., guessing in order of non-increasing probabilities (see Fig. \ref{fig:YZexample})
} Let us now state the main results of this subsection. 

\subsection{Lower and Upper Bounds}

Let $P_X(x)$ to denote the probability distribution of $X$ and define the moments of the guessing cost using a particular guessing function $\mathcal{G}$ as
{
\begin{eqnarray}
\mathbb{E}[C_\mathcal{G}(X)^\rho] = \sum_{i=1}^M P_X({\mathcal{G}^{-1}(i)}) \left[\sum_{j=1}^i c_{\mathcal{G}^{-1}(j)}\right]^\rho \label{orgeqn}
\end{eqnarray}
}
where the costs are not necessarily integers. Let us use the previous notation $c_i = c_{\mathcal{G}^{-1}(i)}$ and define $\mathbf{c}^*=\{c^*_1,c^*_{2},\dots,c^*_M\}$ to be the order of costs obtained by running Algorithm 1 {for a given $\rho > 0$} to find the optimal guessing strategy $\mathcal{G}^*$. This  shall be useful in expressing the lower and upper bounds in the following {two theorems}.

\begin{theorem} \label{thm301}
For any guessing function $\mathcal{G}$, $\rho \geq 0$ and costs $c_j>1$, $\rho$-th moment of the guessing cost is lower bounded by
\begin{align}
\mathbb{E}[C_\mathcal{G}(X)^\rho] \geq \mathbb{E}[C_{\mathcal{G}^*}(X)^\rho] 
\geq \left(\frac{M}{1+\gamma^*}\right)^{-\rho} \exp\left\{\rho H_{\frac{1}{1+\rho}}(X) \right\}  \label{eqn11}
\end{align}
where $\gamma^*$ is the harmonic mean of $\{\sum_j^i c^*_j-1\}'s$ for $i=\{1,2,\dots,M\}$ and $H_{\alpha}(X)$ is Rényi's entropy of order $\alpha$ for a given random variable $X$ { as long as the limit for Renyi's entropy  exists}. 
\end{theorem}

\begin{proof}
The proof of the theorem can be found in {Appendix C}. 
\end{proof}

This lower bound, as will be illustrated in numerical results, is not too tight particularly for large $\rho$. However, this theorem would be useful for asymptotic analysis. For instance, using this result we can demonstrate in the next theorem that the bound given in Theorem \ref{thm301} is tight within a factor of  $\left( M / (1+\gamma^*) \right)^\rho$. 

\begin{theorem} \label{thm31}
For the optimal guessing function $\mathcal{G}^*$, and $\rho \geq 0$, $\rho$-th moment of the  guessing cost is upper bounded by
\begin{eqnarray}
\mathbb{E}[C_\mathcal{G^*}(X)^\rho] \leq  \exp\{ \rho H_{\frac{1}{1+\rho}}(Y)\} \label{eqn2222}
\end{eqnarray}
where  $H_{\alpha}(X)$ is Rényi's entropy of order $\alpha$ for a given random variable $X$. 
\end{theorem}

\begin{proof}
The proof of the theorem can be found in {Appendix D.} 
\end{proof}

\subsection{Relation to Guesswork and Guessing Cost Exponent}

{In this section, we present tight bounds for the logarithm of guessing cost moments, for a series of $M$ random variables variables, which would be useful for our later data storage application. We primarily realize that the }
introduction of a random variable $Z$ is useful for establishing a relationship with the guesswork. From the earlier discussions on the random variable $Z$, we can express a looser lower bound (compared to \eqref{eqn11}) for any guessing function $\mathcal{G}(.)$ by observing the following for $c_j >0$, 
\begin{align}
\mathbb{E}[C_\mathcal{G}(X)^\rho] 
& = \sum_{x} P_X(x) C_{\mathcal{G}}(x)^\rho = \sum_{i=1}^M P_X({\mathcal{G}^{-1}(i)}) \left[\sum_{j=1}^i c_{\mathcal{G}^{-1}(j)}\right]^\rho  \label{eqn100}\\ 
& = \sum_{i=1}^M \sum_{j=1}^{\lfloor c_{\mathcal{G}^{-1}(i)} \rfloor} \frac{P_X(\mathcal{G}^{-1}(i))}{\lfloor c_{\mathcal{G}^{-1}(i)} \rfloor} \left[\sum_{k=1}^{i-1} c_{\mathcal{G}^{-1}(k)} + c_{\mathcal{G}^{-1}(i)}\right]^\rho \\
& \geq   \sum_{i=1}^M \sum_{j=1}^{\lfloor c_{\mathcal{G}^{-1}(i)} \rfloor} \frac{P_X(\mathcal{G}^{-1}(i))}{\lfloor c_{\mathcal{G}^{-1}(i)} \rfloor} \left[\sum_{k=1}^{i-1}\lfloor c_{\mathcal{G}^{-1}(k)} \rfloor+j \right]^\rho  \\ 
& = \mathbb{E}[C_\mathcal{F}(Z)^\rho] \geq \left(1+ \ln \left(\sum_x \lfloor c_x \rfloor \right)\right)^{-\rho} \exp \left\{ \rho H_{\frac{1}{1+\rho}}(Z) \right\} \label{lowerboundalter} 
\end{align}
where the last inequality is due to guesswork and follows directly from \cite{arikan1996} based on the definition of the random variable $Z$. Better lower bounds can be given, however this loose lower bound is enough to prove the following asymptotic result.  Here the guessing function $\mathcal{F}(Z)$ for the random variable $Z$ defined earlier is directly induced from  $\mathcal{G}(X)$. Next, the guessing cost exponent  is given by the following theorem.

\begin{theorem} \label{thmexponent}
{Let $\{X_1,\dots,X_n\}$ be a sequence} of independent random variables  where each is defined over the set $\mathcal{X}_i$ with the associated cost distribution $\mathcal{C}_i${, and random variables $\{Y_i\}, \{Z_i\}$ based on Definition \ref{def31}}. Let $\mathcal{G}^* (X_1,\dots, X_n)$ be an optimal guessing function. Then, for any $\rho>0$, we have
\begin{align}
    \limsup_{n \rightarrow \infty} \frac{1}{n} \ln \left( \mathbb{E}[C_{\mathcal{G}^*}(X_1,X_2,\dots,X_n)^\rho]\right)^{1/\rho} 
    =\mathcal{R}_{\frac{1}{1+\rho}}(\{Y_i\}) \\
        \liminf_{n \rightarrow \infty} \frac{1}{n} \ln ( \mathbb{E}[C_{\mathcal{G}^*}(X_1,X_2,\dots,X_n)^\rho])^{1/\rho} 
    =\mathcal{R}_{\frac{1}{1+\rho}}(\{Z_i\}) 
\end{align}
where $\mathcal{R}_{\frac{1}{1+\rho}}(.)$ denotes the order-$1/(1+\rho)$ Rényi rate {is assumed to exist}, $Y_i$s and $Z_i$s are random variables induced from random variables $X_i$s as defined before. Moreover if the costs are integers, then the limits converge and we will have 
\begin{align}
    \lim_{n \rightarrow \infty} \frac{1}{n} \ln ( \mathbb{E}[C_{\mathcal{G}^*}(X_1,X_2,\dots,X_n)^\rho])^{1/\rho} 
    =\mathcal{R}_{\frac{1}{1+\rho}}(\{X_i\})
\end{align}
\end{theorem}



\begin{proof}
The proof of the theorem can be found in {Appendix E.} 
\end{proof}

These results indicate that the complexity of  guessing cost of a random variable $X$ with strategy $\mathcal{G}$ can be tied to the complexity of guessing two related random variables $Z$ and $Y$ with the induced strategies $\mathcal{F}$ and $\mathcal{H}$, respectively, which are derived based on the cost distribution $\mathcal{C}$ defined earlier. 

\subsection{Improved Bounds: Non-asymptotic regime}

One of the observations is that the provided bounds have the potential for improvement particularly in the non-asymptotic regime similar in spirit to works such as \cite{Bozdas1997}, \cite{sason2018} and \cite{Dragomir98}. These improvements can easily be made after we recognize the relationship between guessing cost and the {standard} guesswork. In the following, we go through these extensions by referring to related past works. We shall also demonstrate how these bounds play out with varying $\rho$.

\subsubsection{Extension of Boztas' bounds \cite{Bozdas1997}}

Let us extend Boztas' upper bound by deriving the analog for the guessing cost. Let us first start with the following definition.
\begin{definition} \label{defbalcost}
For a given random variable $X$ and $\rho > 0$, the \textit{balancing cost} $\overline{c}_X(\rho)$ is defined to satisfy the following equality
\begin{eqnarray}
\sum_{i=1}^M \left( \sum_{j=1}^{i-1} c_j \right)^\rho p_i = \sum_{i=1}^M \left( \sum_{j=1}^{i} c_j - \overline{c}_X(\rho) \right)^\rho p_i \label{eqn62}
\end{eqnarray}
and equals a constant if costs are fixed i.e., $c_1=\dots=c_M=c$ for some constant $c \in \mathbb{R}$.
\end{definition}

\begin{remark}
Note that for the special case $\rho = 1$, we will have $ \overline{c}_X(1) = \sum_i c_ip_i$, i.e., balancing cost would be equivalent to the average (expected) cost of guessing. 
\end{remark}

Now considering telescoping sequence argument for $\rho \geq 1$, we observe the following relation
\begin{align}
\sum_{i=1}^M \left[\left( \sum_{j=1}^{i} c_j \right)^\rho - \left( \sum_{j=1}^{i-1} c_j \right)^\rho\right] \frac{p_i}{\lceil c_i \rceil} &\leq  \sum_{i=1}^M \left[\left( \sum_{j=1}^{i} \lceil c_j \rceil \right)^\rho - \left( \sum_{j=1}^{i-1} \lceil c_j \rceil \right)^\rho\right] \frac{p_i}{\lceil c_i \rceil}  \\
&=\sum_{i=1}^M \sum_{z=1}^{\lceil c_i \rceil} \left( \Big(\sum_{k=1}^{i-1}\lceil c_k \rceil+z\Big)^\rho - \Big(\sum_{k=1}^{i-1}\lceil c_k \rceil+z-1\Big)^\rho \right) \frac{p_i}{\lceil c_i \rceil} 
\end{align}
Finally using the equality provided in Eqn. \eqref{eqn62}, we get
\begin{align} 
 \mathbb{E}[C_\mathcal{G}(X)^\rho] -  \mathbb{E}[(C_\mathcal{G}(X)-\overline{c}_X(\rho))^\rho] & = \sum_{i=1}^M \left[\left( \sum_{j=1}^{i} c_j \right)^\rho - \left( \sum_{j=1}^{i} c_j - \overline{c}_X(\rho) \right)^\rho \right] p_i\label{eqn622} \\
 &= \sum_{i=1}^M  \left[\left( \sum_{j=1}^{i} c_j \right)^\rho - \left( \sum_{j=1}^{i-1} c_j \right)^\rho\right] p_i \\
 & \leq \sum_{i=1}^M \sum_{z=1}^{\lceil c_i \rceil} \left( \Big(\sum_{l=1}^{i-1}\lceil c_l \rceil+z\Big)^\rho - \Big(\sum_{l=1}^{i-1}\lceil c_l \rceil+z-1\Big)^\rho \right) p_i \\
 & = \sum_{k=1}^{M^\prime} (k^\rho - (k-1)^\rho) q_k \leq \left[ \sum_{k=1}^{M^\prime} q_k^{1/\rho}\right]^\rho \label{eqn63}
\end{align}
where $M^\prime = \sum_{i=1}^M \lceil c_i \rceil$ and \begin{eqnarray}
q_k = p_i \textrm{ for } \sum_{l=1}^{i-1} \lceil c_l \rceil < k \leq  \sum_{l=1}^{i} \lceil c_l \rceil \textrm{ and } i = 1,\dots,M,
\end{eqnarray}
{
\begin{eqnarray}
q_{k+1}^{1/\rho} \leq \frac{1}{k} (q_1^{1/\rho}+\dots+q_{k}^{1/\rho}), \textrm{  for } k=1,\dots, M^\prime-1.\label{eqn64}
\end{eqnarray}
Note that the inequality in \eqref{eqn63} follows  from the Lemma in  \cite{Bozdas1997} as long as the ``weights'' $q_1,\dots,q_{M^\prime}$ are non-negative reals satisfying the inequality  given in \eqref{eqn64}. }
We can show that the necessary condition for optimal strategy derived earlier will satisfy this inequality. Hence, this is a looser condition making the inequality apply to a broader range of guessing functions other than the optimal. Note here that $\sum_k q_k = \sum_i \lceil c_i \rceil p_i \not= 1$ unless $c_i=1$ for all $i=1,\dots,M$ i.e., $q_k$s are not forming a probability distribution for non-unity costs.
Next, let us provide our theorem as an extension/generalization of Boztas' arguments.

\begin{theorem}\label{theo:boztasUpperBoundupdatedtheo3_4}
For $\overline{c}_X(.)$  as given in Definition 3.3 and all guessing functions $\mathcal{G}$ for a random variable $X$ inducing $\{q_k\}$s which satisfy the relation in $\eqref{eqn64}$ for $\rho = m + 1$ where $m \geq 1$ is an integer, the $m$-th moment of the guessing cost can be upper bounded by the recursive relation
\begin{align}
\mathbb{E}[C_\mathcal{G}(X)^{m}] &\leq \frac{1}{ \overline{c}_X(m+1)(m+1)} \left[ \left[\sum_{k=1}^{M^\prime} q_k^\frac{1}{m+1}\right]^{m+1} + \sum_{l=0}^{m-1} \binom{m+1}{l} \mathbb{E}[C_\mathcal{G}(X)^{l}] (-\overline{c}_X(m+1))^{m+1-l} \right]
\end{align}
where $m\geq1$ is a positive integer and $M^\prime = \sum_{i=1}^M \lceil c_i \rceil$. 
\end{theorem}

\begin{proof}
Using equations (\ref{eqn62}), (\ref{eqn63}) and the Binomial expansion, we have the following inequalities for integer $m$,
\begin{align}
 \mathbb{E}[C_\mathcal{G}(X)^{m+1}] -  \mathbb{E}[(C_\mathcal{G}(X)-\overline{c}_X(m+1))^{m+1}] &=  
 \mathbb{E}[C_\mathcal{G}(X)^{m+1}] - \sum_{l=0}^{m+1} \binom{m+1}{l} \mathbb{E}[C_\mathcal{G}(X)^{l}] (-\overline{c}_X(m+1))^{m+1-l} \\ 
 &\leq \left[\sum_{k=1}^{M^\prime} q_k^\frac{1}{m+1}\right]^{m+1} 
 \end{align}
 which implies that
 \begin{align}
 \overline{c}_X(m+1)(m+1) \mathbb{E}[C_\mathcal{G}(X)^{m}] &\leq \left[\sum_{k=1}^{M^\prime} q_k^\frac{1}{m+1}\right]^{m+1}  + \sum_{l=0}^{m-1} \binom{m+1}{l} \mathbb{E}[C_\mathcal{G}(X)^{l}] (-\overline{c}_X(m+1))^{m+1-l}
\end{align}
from which the result follows.
\end{proof}

The main difference of our result compared to that of Boztas is the introduction of $\{q_k\}$s and the term $\overline{c}_X(m+1)$. In case of $m=1$, we would have
\begin{eqnarray}
\mathbb{E}[C_\mathcal{G}(X)] \leq \frac{1}{2\overline{c}_X(2)}  \left[\sum_{k=1}^{M^\prime} q_k^\frac{1}{2}\right]^{2} + \frac{\overline{c}_X(2)}{2} 
\end{eqnarray}
subject to $q_{k+1}^{1/2} \leq \frac{1}{k} (q_1^{1/2}+\dots+q_{k}^{1/2}), \textrm{  for } k=1,\dots, M^\prime-1$. Similarly for $m=2$, we shall have
\begin{align}
\mathbb{E}[C_\mathcal{G}(X)^2] &\leq \frac{1}{3\overline{c}_X(3)}  \left[\sum_{k=1}^{M^\prime} q_k^\frac{1}{3}\right]^{3}  + \overline{c}_X(3)\mathbb{E}[C_\mathcal{G}(X)] - \frac{{\overline{c}_X(3)}^2}{3} \\
&\leq  \frac{1}{3\overline{c}_X(3)}  \left[\sum_{k=1}^{M^\prime} q_k^\frac{1}{3}\right]^{3}  + \frac{\overline{c}_X(3)}{2\overline{c}_X(2)} \left[\sum_{k=1}^{M^\prime} q_k^\frac{1}{2}\right]^{2} + \overline{c}_X(3)\left(\frac{\overline{c}_X(2)}{2}- \frac{{\overline{c}_X(3)}}{3}\right)
\end{align} 
subject to the conditions $q_{k+1}^{1/2} \leq \frac{1}{k} (q_1^{1/2}+\dots+q_{k}^{1/2})$ and $q_{k+1}^{1/3} \leq \frac{1}{k} (q_1^{1/3}+\dots+q_{k}^{1/3}), \textrm{  for } k=1,\dots, M^\prime-1$.

We finally note that these expressions/bounds form a generalization of Boztas' results and requires the calculation of the balancing cost for integer $\rho$s. We provide a gradient descent scheme in Algorithm
\ref{alg:GradientDescentBarc} for efficiently finding the balancing cost for a given integer $\rho$.


\begin{algorithm}[!t]
\begin{algorithmic}[1] 
\State \textbf{function}  \textbf{gradientDescent} {$\mathbf{\bar{c}}$}($\mathbf{p}, ~\mathbf{c},~n,~\delta,~\mu,\rho $) 
\State $M \gets |\mathbf{p}|$ 
\State \textit{minusCost} $\gets \sum_{i=1}^{M}C_{i-1}^{\rho}p_i $ where $C_{i-1}=\sum_{j=1}^{i-1}c_j $
\State $minusCost'(\bar{c}) \gets \sum_{i=1}^{M}{(C_{i}-\bar{c})}^{\rho}p_i $ where $C_{i}=\sum_{j=1}^{i}c_j $
\State $f'(\bar{c}) \gets  |\sum_{i=1}^{M}{(C_{i-1}}^{\rho}p_i) - ( {{\bar{c}(-\rho)}(C_{i}-\bar{c})^{\rho-1}p_i +(C_{i}-\bar{c})}^{\rho}p_i) |$ where $C_{i}=\sum_{j=1}^{i}c_j$
\State $\bar{c} = \min(\mathbf{c})$
\For{$i=1$ \textbf{to} $n$} \Comment{$n$ represents the iteration count}
\State  \textit{step} $\gets -\delta \times f'(\bar{c})$ \Comment{$\delta$ represents  the step size}
\If {\textit{step} $< 0$}
\State $\bar{c} = \bar{c} - \textit{step} $ 
\Else
\State $\bar{c} = \bar{c} + $ \textit{step}
\EndIf
\If {$|\textit{step}|\leq \mu \vee \bar{c}\geq max(\mathbf{c}) $} \Comment{$\mu$ represents  the step tolerance}
\State \textbf{return} $\{ \bar{c}$\}
\EndIf
\EndFor
\State \textbf{return} $\{-1\}$ \Comment{Notify an error}
\end{algorithmic}
\caption{} \label{alg:GradientDescentBarc}
\end{algorithm}

\subsubsection{Extension of Sason's bounds \cite{sason2018}}  In particular, we have the following improved lower bounds for any guessing strategy $\mathcal{G}$ and $\rho > 0$ that show better performance in the non-asymptotic regime of $\rho$,
\begin{align}
\mathbb{E}[C_\mathcal{G}(X)^\rho] \geq \mathbb{E}[C_\mathcal{F}(Z)^\rho] &\geq \sup_{\beta \in (-\rho,\infty)-\{0\}} \exp\left\{\frac{\rho}{\beta}\left[H_{\frac{\beta}{\beta+\rho}}(Z) - \log u_{\sum_x \lfloor c_x \rfloor} (\beta) \right]\right\} \\
&= \sup_{\beta \in (-\rho,\infty)-\{0\}} \left[u_{\sum_x \lfloor c_x \rfloor} (\beta)\right]^{-\frac{\rho}{\beta}}\exp\left(\frac{\rho}{\beta}H_{\frac{\beta}{\beta+\rho}} (Z)\right) \label{sasonlb}
\end{align}
where {$u_{|\mathcal{Z}|}(\beta)$} is defined {similarly as} in \cite{sason2018} and given by  
\begin{equation} 
u_{|\mathcal{Z}|}(\beta)=\begin{cases}
\ln |\mathcal{Z}| +\gamma+ \frac{1}{2|\mathcal{Z}|}-\frac{5}{6(10{(|\mathcal{Z}|)}^2+1)} & \beta = 1  \\
\min\{\zeta(\beta)-\frac{(|\mathcal{Z}|+1)^{1-\beta}}{\beta-1}- \frac{(|\mathcal{Z}|+1)^{1-\beta}}{2}, u_{|\mathcal{Z}|}(1)\} & \beta >1\\
1+\frac{1}{1-\beta}\left[{(|\mathcal{Z}|+\frac{1}{2})}^{1-\beta}- \left(\frac{3}{2}\right)^{1-\beta} \right] & |\beta|<1\\
\frac{{(|\mathcal{Z}|)}^{1-\beta}-1}{1-\beta}+\frac{1}{2}(1+{|\mathcal{Z}|}^{-\beta}) & \beta \leq -1
\end{cases}
\end{equation}
where $ |\mathcal{Z}| =\sum_x \lfloor c_x \rfloor$, $ \gamma \approx 0.5772$ is the Euler-Mascheroni constant and $\zeta(\beta) = \sum_{n=1}^\infty \frac{1}{n^\beta}$ is the Riemann zeta function for $\beta > 1$. Here the first inequality follows due to equations \eqref{eqn100}--\eqref{lowerboundalter}. Moreover the second inequality follows due to guesswork arguments given in  \cite{sason2018} which are directly applicable to random variable $Z$ as its cost distribution assumes only unity values. As an extension of the upper bound, we provide the following theorem. 
\begin{theorem} \label{sasonbound1}
For any guessing function $\mathcal{G}$, $\rho \geq 0$ and costs $c_j>1$ associated with $\{q_k\}$s for a random variable $X$ satisfying both $q_{k+1}^{\frac{1}{\rho}} \leq \frac{1}{k} (q_1^{\frac{1}{\rho}}+\dots+q_{k}^{\frac{1}{\rho}})$ and  $q_{k+1}^{\frac{1}{1+\rho}} \leq \frac{1}{k} (q_1^{\frac{1}{1+\rho}}+\dots+q_{k}^{\frac{1}{1+\rho}})$ for $ k=1,\dots, M^\prime-1$, then the $\rho$-th moment of the guessing cost is upper bounded  by
\begin{align}
        \mathbb{E}[C_\mathcal{G}(X)^\rho] \leq \frac{1}{\overline{c}_{{min}_X}(\rho)(1+\rho)} \left[ \sum_{k=1}^{M^\prime} q_k^{\frac{1}{1+\rho}}\right]^{1+\rho} + \overline{c}_{{min}_X}^{\rho \mathbbm{1}_{\rho < 1}}(\rho) \left[ \sum_{k=1}^{M^\prime} q_k^{1/\rho}\right]^{\rho \mathbbm{1}_{\rho \geq 1}} - \frac{\overline{c}^\rho_{{min}_X}(\rho)}{1+\rho}
\end{align}
where $\overline{c}_{{min}_X}(\rho) = \min\{\overline{c}_X(\rho),\overline{c}_X(1+\rho)\}$, {$M^\prime = \sum_{i=1}^M \lceil c_i \rceil$}, $\mathbbm{1}_A$ is the indicator function and equals 1 if the condition $A$ is true otherwise 0, and $\overline{c}_X(\rho)$, $\overline{c}_X(1+\rho)$ are as defined before for a given $\rho$ and can be found using Algorithm \ref{alg:GradientDescentBarc}.  
\end{theorem}

\begin{proof}
The proof of the theorem can be found in {Appendix F.} 
\end{proof}

\begin{remark}
Theorem \ref{sasonbound1} may be loose for a given parameter set compared to previous upper bounds. However, we note that Theorem \ref{sasonbound1} is in similar form to Theorem \ref{theo:boztasUpperBoundupdatedtheo3_4} except it is non-recursive and assumes any real $\rho \geq 0$ rather than an integer. 
\end{remark}


{It will become evident that by utilizing the claim given above, we will be able to improve the upper bound, particularly for values of $\rho$ that are relatively small. Additionally, it is worth noting that the subsequent theorem provides an opportunity to refine this bound even further, specifically for $\rho \in (0,2]$.}

\begin{theorem} \label{Thm36}
For any guessing function $\mathcal{G}$ and the cost of guessing $C_\mathcal{G}(.)$, $\rho \in (0,2]$ and costs $c_j>1$ associated with $\{q_k\}$s for a random variable $X$ satisfying both $q_{k+1}^{\frac{1}{\rho}} \leq \frac{1}{k} (q_1^{\frac{1}{\rho}}+\dots+q_{k}^{\frac{1}{\rho}})$ and  $q_{k+1}^{\frac{1}{1+\rho}} \leq \frac{1}{k} (q_1^{\frac{1}{1+\rho}}+\dots+q_{k}^{\frac{1}{1+\rho}})$ for $ k=1,\dots, M^\prime-1$, then the $\rho$-th moment of the guessing cost is upper bounded  by
\[
  \mathbb{E}[C_\mathcal{G}(X)^\rho] \leq \left.
  \begin{cases}
    \frac{1}{\overline{c}_{{min}_X}(\rho) (1+\rho)}  \left[ \sum_{k=1}^{M^\prime} q_k^{\frac{1}{1+\rho}}\right]^{1+\rho} +  \frac{\rho \overline{c}^\rho_{{min}_X}(\rho)}{1+\rho} P(\overline{c}_{{min}_X}(\rho) \leq C_\mathcal{G}(X) < \overline{c}_{{min}_X}(\rho)+1) \\ + \left((\overline{c}_{{min}_X}(\rho)+1)^\rho - \frac{(\overline{c}_{{min}_X}(\rho)+1)^{1+\rho}-1}{\overline{c}_{{min}_X}(\rho)(1+\rho)}\right) P(C_\mathcal{G}(X) \geq \overline{c}_{{min}_X}(\rho)+1)
    & \text{for } \rho \in (0,1) \\
    \frac{1}{1+\rho}  \left[ \sum_{k=1}^{M^\prime} q_k^{\frac{1}{1+\rho}}\right]^{1+\rho}  + \frac{1}{\rho} \left[ \sum_{k=1}^{M^\prime} q_k^{\frac{1}{\rho}}\right]^{\rho} + \frac{\overline{c}_{{min}_X}^\rho(\rho)(\rho^2-\overline{c}_{{min}_X}(\rho)\rho-1)}{\rho(1+\rho)}
    & \text{for } \rho \in [1,2]
  \end{cases}
  \right.
\]
where $\overline{c}_{{min}_X}(\rho) = \min\{\overline{c}_X(\rho),\overline{c}_X(1+\rho)\}$ {and $M^\prime = \sum_{i=1}^M \lceil c_i \rceil$}. 
\end{theorem}
\begin{proof}
The proof of the theorem can be found in {Appendix G.} 
\end{proof}

Furthermore, we provide the following recursive upper bound that can be used along with Theorem \ref{Thm36} to extend the previous result to explicit upper bounds for larger values of $\rho > 2$. 

\begin{theorem} \label{Thm37}
For any guessing function $\mathcal{G}$ and the cost of guessing $C_\mathcal{G}(.)$, $\rho \in (2,\infty)$ and costs $c_j>1$ associated with $\{q_k\}$s for a random variable $X$ satisfying $q_{k+1}^{\frac{1}{1+\rho}} \leq \frac{1}{k} (q_1^{\frac{1}{1+\rho}}+\dots+q_{k}^{\frac{1}{1+\rho}})$ for $k=1,\dots, M^\prime-1$, $\rho$-th moment of the guessing cost is upper bounded  by
\begin{align}
        \mathbb{E}[C_\mathcal{G}(X)^\rho] \leq \frac{1}{\overline{c}_X(1+\rho)(1+\rho)} \left[ \sum_{k=1}^{M^\prime} q_k^{\frac{1}{1+\rho}}\right]^{1+\rho}  + \frac{\rho \overline{c}_X(1+\rho)}{2} \mathbb{E}[C_\mathcal{G}(X)^{\rho-1}] - \frac{\rho(\rho-1)}{2(1+\rho)} \label{eqn399}
\end{align}
\end{theorem}
\begin{proof}
The proof of the theorem can be found in {Appendix H.} 
\end{proof}

\begin{remark} \label{remark32}
Using Theorem \ref{Thm37}, we can find an explicit bound for $\rho$'s satisfying $i+1 \geq \rho > i$ for all integers $i > 2$. We can obtain these bounds by applying Equation \eqref{eqn399} for $i-2$ times using the result of Theorem \ref{Thm36}. In this case however, the set of conditions would be more restrictive i.e., we would require to satisfy $q_{k+1}^{\frac{1}{1+\rho}} \leq \frac{1}{k} (q_1^{\frac{1}{1+\rho}}+\dots+q_{k}^{\frac{1}{1+\rho}})$ for $\rho,\rho-1,\dots,\rho-i$ and $k=1,\dots, M^\prime-1$ all at the same time. 
\end{remark}

In order to help understand Remark \ref{remark32} with an example, let us consider for instance $\rho \in (2,3]$. In this case we can apply the result of Theorem \ref{Thm36} to get
\begin{align}
    \mathbb{E}[C_\mathcal{G}(X)^{\rho-1}] \leq \frac{1}{\rho}  \left[ \sum_{k=1}^{M^\prime} q_k^{\frac{1}{\rho}}\right]^{\rho}  + \frac{1}{\rho-1} \left[ \sum_{k=1}^{M^\prime} q_k^{\frac{1}{\rho-1}}\right]^{\rho-1} + \frac{\overline{c}_{{min}_X}^{\rho-1}(\rho-1)((\rho-1)^2-\overline{c}_{{min}_X}(\rho-1)(\rho-1)-1)}{(\rho-1)\rho}
\end{align}
which is subject to $q_{k+1}^{\frac{1}{\rho}} \leq \frac{1}{k} (q_1^{\frac{1}{\rho}}+\dots+q_{k}^{\frac{1}{\rho}})$ and $q_{k+1}^{\frac{1}{\rho-1}} \leq \frac{1}{k} (q_1^{\frac{1}{\rho-1}}+\dots+q_{k}^{\frac{1}{\rho-1}})$. Then using Theorem \ref{Thm37} the upper bound for $\rho \in (2,3]$ can be expressed in a closed form as 
\begin{align}
    \mathbb{E}[C_\mathcal{G}(X)^\rho] &\leq \frac{1}{\overline{c}_X(1+\rho)(1+\rho)} \left[ \sum_{k=1}^{M^\prime} q_k^{\frac{1}{1+\rho}}\right]^{1+\rho} - \frac{\rho(\rho-1)}{2(1+\rho)} \nonumber \\ 
    & \ \ \ + \overline{c}_X(1+\rho) \left[ \frac{1}{2}  \left[ \sum_{k=1}^{M^\prime} q_k^{\frac{1}{\rho}}\right]^{\rho} + \frac{\rho}{2(\rho-1)} \left[ \sum_{k=1}^{M^\prime} q_k^{\frac{1}{\rho-1}}\right]^{\rho-1} + \frac{\overline{c}_{{min}_X}^{\rho-1}(\rho-1)(\rho^2-2\rho-\overline{c}_{{min}_X}^{\rho-1}(\rho-1)(\rho-1))}{2(\rho-1)} \right]  
\end{align}
with the additional constraint $q_{k+1}^{\frac{1}{1+\rho}} \leq \frac{1}{k} (q_1^{\frac{1}{1+\rho}}+\dots+q_{k}^{\frac{1}{1+\rho}})$ for $k=1,\dots, M^\prime-1$. 

\begin{remark}
It is not hard to verify that the bounds given in Theorems \ref{sasonbound1}, \ref{Thm36} and \ref{Thm37} will be reduced to Sason's bounds given in \cite{sason2018} if we assume constant and unit costs. Hence these bounds are useful extensions and characterize a more general scenario. 
\end{remark}

\subsection{Extension of Dragomir's bounds \cite{Dragomir98}}

Finally, we would like to remark on the Dragomir's bounds which was originally presented in the context of guesswork. These bounds have been introduced right after Boztas' bounds are published \cite{Dragomir98}. Unfortunatelly these bounds are quite loose particularly in the context of guessing cost. The proposed bounds were based on the following theorem. 
\begin{theorem}
Let $a_i,b_i \in \mathbb{R}$ for $i \in {[n]}$ such that
\begin{align}
    a_{min} \leq a_i \leq a_{max}, \ \ b_{min} \leq b_i \leq b_{max} \ \ \text{ for all } \ \ i =1,\dots,n
\end{align}
with $a_{max} = \min\{a_i\}$ and $b_{max} = \min\{b_i\}$.  Then, we have the inequality
\begin{align}
    \left| \frac{1}{n}\sum_{i=1}^n a_i b_i - \left(\frac{1}{n} \sum_{i=1}^n a_i \right)\left(\frac{1}{n} \sum_{i=1}^n b_i \right) \right| \leq \frac{1}{4}(a_{max}-a_{min})(b_{max}-b_{min}) \label{eqn72}
\end{align}
\end{theorem}
\begin{proof}
The proof can be found in \cite{Dragomir98}.
\end{proof}
Let $a_i=f_i^\rho=\left(\sum_{j=1}^ic_j\right)^\rho$ and $b_i=p_i$ in equation (\ref{eqn72}). Also, let us define random variable $U$ with exactly the same cost distribution $\mathcal{C}$ of $X$ and uniform probability distribution, then for any gussing strategy $\mathcal{G}$ we have
\begin{align}
    \left|  \mathbb{E}[C_\mathcal{G}(X)^\rho] - \mathbb{E}[C_\mathcal{G}(U)^\rho]  \right| \leq \frac{M(p_{max}-p_{min})}{4}\left(\left(\sum_{j=1}^Mc_j\right)^\rho-c_{min}^\rho\right)
\end{align}
where $c_{min} = \min\{c_i\}$. Note that this relation defines both an upper and a lower bound for $\mathbb{E}[C_\mathcal{G}(X)^\rho]$. The bound can be tightened using the optimal guessing strategy $\mathcal{G}^*$. However, Dragomir's bounds are generally looser compared to {that of} Sason's and hence we omit to present numerical results for this bound.

\subsection{Numerical Results}

First, let us provide several numerical results to be able to illustrate how close the provided bounds are for finite values of costs, $\rho$ and $M$. The exact moments for the optimal guessing strategy are calculated using Algorithm 1 and denoted by OPT. The results are provided in Fig. \ref{fig:costofguessingresults}. More specifically, inspired from the past research \cite{sason2018}, we consider the quantity $\frac{1}{\rho} \ln\mathbb{E}[{C_{\mathcal{G}^*}(X)}^{\rho}]$ in our comparisons where $\rho \in [0.25,10]$. The probability of each choice is generated using geometric distribution as assumed in \cite{sason2018} with the restricted probability distribution $P_X(x) = (1-a)a^{x-1}/(1-a^M)$ using $M=32$ and the parameter $a=0.9$.  The non-integer cost values are generated based on a truncated normal distribution defined in the range $(1,100)$ with the same mean and variance i.e., $\mu=\sigma^2=16$.

\begin{figure}[!t] 
	\centering
\includegraphics[width=\textwidth]{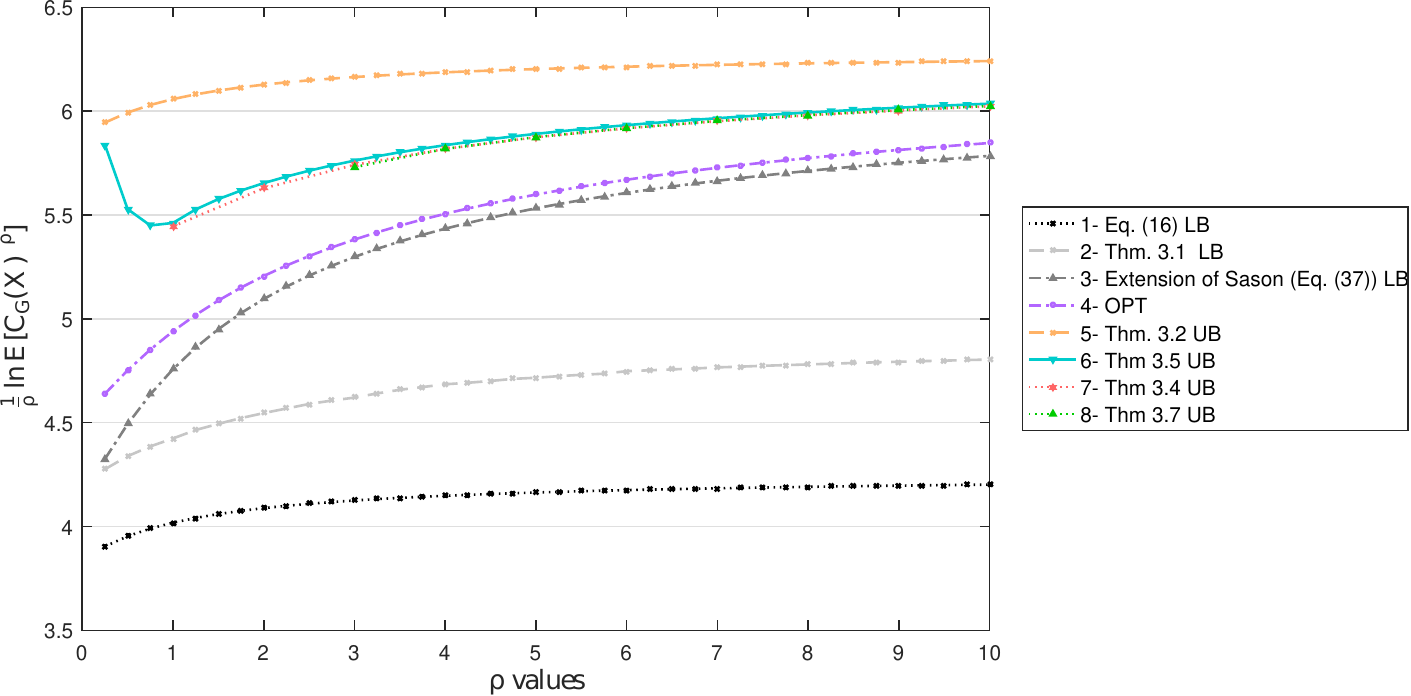} 
\centering \caption{This figure presents the exact value as well as the lower and upper bounds for $\frac{1}{\rho} {\ln\mathbb{E}[{C_{\mathcal{G}^*}(X)}^{\rho}] }$.}\label{fig:costofguessingresults}
\end{figure}

As shown in Fig. \ref{fig:costofguessingresults}, the closest values to   $\frac{1}{\rho} \ln\mathbb{E}[{C_{\mathcal{G}^{*}}(X)}^{\rho}]$ for $\rho \in (0.25,10]$, are given by the Eq. (\ref{sasonlb}), which are followed by the bounds provided in Theorem \ref{thm301} and Eq.(\ref{lowerboundalter}). On average, the lower bounds of $\frac{1}{\rho} \ln\mathbb{E}[{C_{\mathcal{G}^{*}}(X)}^{\rho}]$ using Eq. (\ref{sasonlb}) is $16.3\%$ and $30.5\%$  higher than that of bounds due to  Theorem \ref{thm301} and Eq. (\ref{lowerboundalter}), respectively. In fact, it is interesting to show that bounds of Theorem \ref{thm301} and Eq. (\ref{lowerboundalter}) are not asymptotically tight. The tightest bound is achieved by the bound given in Theorem \ref{theo:boztasUpperBoundupdatedtheo3_4} among other alternative  upper bounds. The bounds given in  Theorem \ref{theo:boztasUpperBoundupdatedtheo3_4} are $5.98\%$ and $0.391\%$ less than the bounds given in Theorem \ref{thm31} and Theorem \ref{sasonbound1} for $\rho \in \{1,2,\hdots,10\}$, respectively. Moreover, for $\rho \in \{4,5,\hdots,10\}$ the bound values of Theorem \ref{theo:boztasUpperBoundupdatedtheo3_4} are $0.023\%$ less than that of Theorem \ref{Thm37}. Notice also that bounds given in Theorem \ref{theo:boztasUpperBoundupdatedtheo3_4} and 
Theorem \ref{Thm37} are only valid for integer values of $\rho$ and Theorems \ref{Thm36}  and \ref{Thm37} are complementary and should be considered together. 



 \section{An Application: Distributed Data Regeneration}


{In this section, we provide an application of the guessing cost within the context of a distributed data storage in which data
content regeneration and repair are necessary to maintain the data durability. Such a data repair application scenario involves a
slight variation of the guessing cost problem (introduced earlier), which is shown to be quite useful in this section in deriving
optimal protocol design for highly dynamic  networks, {for instance, wireless networks or mobile ad-hoc networks.}
 
 \subsection{Long Block Length Sparse Graph Codes With A Back-up Master}
 
Let us consider a cellular network with a master-slave configuration for a distributed data storage scenario in which the data protection is provided by a long block length $(n,k)$ sparse graph code. Each slave node in the system is assumed to store a single coded symbol of a \textit{codeword}. In addition, a master node constitutes a backup system (a.k.a. a base station) and keeps the copy of all coded symbols. If one of the slave nodes fails, departs the cellular network, or becomes permanently unavailable, {it is interpreted as} loss of a coded symbol in the system. Thanks to the multiple check relations defined for that lost symbol in the sparse graph code, there would be multiple options of repair for that specific node. To be able to maintain instantaneous reliability, this symbol is required to be repaired as soon as possible. 
 
In a highly dynamic network \cite{Le2017}, it may not be possible to obtain the status of all nodes (due to other unexpected failures or network link breakages and congestion) instantaneously, or else it may be time and bandwidth costly to contact the master directly and retrieve that information. Therefore, in that case, the newcomer node needs to adapt the best guessing strategy and choose among the multiple repair options to complete the repair process (either exactly or functionally) as quickly as possible using minimum network resources.
 
\begin{figure}[t!]
\centering
  \includegraphics[width=0.75\linewidth]{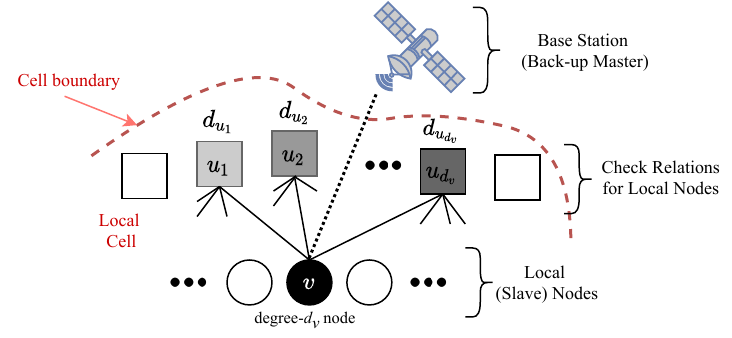}
  \caption{An example repair process using an LDPC code Tanner graph. $d_v$ represents the degree number of the lost symbol/node $v$ whereas the $d_{u_1},\dots,d_{u_{d_v}}$ are the degrees of the potential repair check relations.}
  \label{fig:sgc2}
  \vspace{-4mm}
\end{figure}
 
Let us suppose one of the degree-$d_v$ symbols of an irregular LDPC code, shown as a black-colored node in Fig. \ref{fig:sgc2} is to be exactly repaired. Suppose it is connected to check nodes of degrees $d_{u_1},d_{u_2},\dots,d_{u_{d_v}}$, as shown in the same figure. Accordingly, let us define the costs associated with each repair option to be the number of downloaded symbols, i.e., $c_j\triangleq d_{u_j}-1$\footnote{Here, due to large block length assumption, it is assumed that subsequent guesses cannot help each other. In addition, the cost of download can also be scaled with the link weight for a more realistic communication scenario. In an alternative context, the physical distances between nodes could have also been part of this cost definition, making the rest of our discussion more general.} for all $j$ satisfying $1 \leq j \leq M-1$ with $M=d_v+1$  i.e., each symbol download within the same cell has a unit cost.  One of the differences of this application scenario compared to the standard cost of guessing is that the probabilities are functions of costs as will be explored next. The following proposition establishes a condition for contacting the master node under optimal guessing context and independent node loss model. 

\begin{theorem} \label{prop2}
 Let each slave node to be independently unavailable/failed with probability $q > 0$. Assuming a degree-$d_v$ node is lost, let also $c_M$ be the cost of contacting the back-up node and $c_{max} \triangleq \max\{c_1,c_2,\dots,c_{d_v}\}$ satisfying 
 \begin{eqnarray}
 c_M \geq c_{max}((1-q)^{-c_{max}}-1)  \geq c_{max} \label{eqprop41}
 \end{eqnarray}
 where $M=d_v+1$. Then guessing check relations  as well as the back-up master   in the order of non-decreasing costs minimizes the average cost of downloaded symbols in the node repair process. {(Here we use the guessing term  for trying these relations  until  the lost symbol is repaired or using back-up master if this symbol could not be repaired using local nodes' check relations.)}
\end{theorem} 
\begin{proof}
Assuming independence, the probability that $j$-th check node will successfully repair the gray-colored node of Fig. \ref{fig:sgc2} can be shown to be of the form
 \begin{align}
p_j = (1-q)^{c_j} \prod_{i=1}^{j-1} (1-(1-q)^{c_i})  \textrm{  with  }  p_{M} = 1 - \sum_{j=1}^{d_v} p_j = \prod_{i=1}^{M-1} (1-(1-q)^{c_i}) \label{eqn43}
 \end{align}
from which we realize that the probabilities are dependent on the costs. In a more general version of the problem, the costs of the check nodes may take values independent of the degrees (e.g., the communication cost  required for obtaining a variable node may be different). In search of an optimal strategy, we need to think about $p_j$'s and $c_j$'s at the same time. Fortunately from equation (\ref{eqn43}), we can express $p_j$'s recursively for $j \leq M-1$,
\begin{eqnarray}
p_j = p_{j-1} \left[(1-q)^{c_j-c_{j-1}}-(1-q)^{c_j}\right]
\end{eqnarray}
which implies that if $c_{j-1} \leq c_j$, due to $0 < (1-q)^s \leq (1-q)^t \leq 1$ for all positive $t \leq s$ and $q\in (0,1)$, we shall have $p_{j} \leq p_{j-1}$. Therefore, rearranging costs in non-decreasing order leads to rearrangement of probabilities in non-increasing order. But this result implies that the necessary condition of Remark \ref{prop1} for $\rho=1$ i.e., $c_ip_j \leq c_jp_i$ is satisfied for all $i,j \in \{1,2,\dots,M-1\}$ and $i \leq j$. Note that if $c_{j-1} > c_j$ was the case, we would not be able to satisfy the necessary condition. 
In order to contact the master (back-up) node when no neighboring nodes are able to help, we then have to satisfy the necessary condition $c_Mp_{M-1} \geq c_{M-1}p_M$. Using equation \eqref{eqn43}, this condition can be reexpressed as
\begin{equation}
   c_M(1-q)^{c_{M-1}}\geq c_{M-1}(1-(1-q)^{c_{M-1}}) \label{inequal1}
\end{equation}
which is the upper bound in inequality \eqref{eqprop41} with $c_{max} = c_{M-1}$. We finally recognize that lower bound  in  inequality \eqref{eqprop41} while costs satisfying $c_i > 1$ means $(1-q)^{c_{max}} \leq 1/2$\footnote{This implies an upper bound on $q$ i.e., $q \geq 1-2
^{-1/c_{max}}$.}. Considering it with inequality in \eqref{inequal1}, this condition reduces to $c_M \geq c_{M-1}$ which completes the proof of the optimality of the non-decreasing cost order. We finally note that the lower bound inequality need not be satisfied for the back-up master to be the last resort. In fact, the upper bound inequality is a sufficient condition for that. However $c_M \geq c_{M-1}$ becomes only necessary if the lower bound inequality is satisfied and hence the assertion of the theorem follows. 
\end{proof}


Let us associate the random variable $X_v$ with a variable node $v$ (having degree-$d_v$) that characterizes the identification of the right check node for a successful repair. Note that a specific node failure pattern determines usable options of repair for that variable node. For instance given the finite set $\mathcal{X}_u = \{u_1,u_2,\dots,u_{d_v}\}$ associated with the costs $\mathcal{C}_u = \{d_{u_1}-1,d_{u_2}-1,\dots,d_{u_{d_v}}-1\}$,  $X_v = u_j$  indicates that the check relation $u_j$ would be the first option for repair (i.e., $\mathcal{G}(X_u=u_j)=1$) if $d_{u_j} \leq d_{u_i}$ for all $i \in [d_v], i \not= j$ (due to proposition 4.1) for a successful regeneration.  Furthermore, let $\mathcal{G}^*(X_1,X_2,\dots,X_n)$  denote the optimal guessing function for the value of a joint realization of independent random variables $X_1,X_2,\dots,X_n$.  Then due to Theorem \ref{thmexponent}, for large enough block length (number of nodes $n$ tends large), the moments of repair bandwidth (cost in terms of downloaded symbols) using the optimal guessing strategy can be well approximated by the Rényi entropy rate (with equality in the ideal case),
\begin{eqnarray}
\mathbb{E}[C_{\mathcal{G}^*}(X_1,X_2,\dots,X_n)^\rho] \thickapprox \prod_{{i=1}}^{{n}} \exp \{ \rho H_{\frac{1}{1+\rho}}(X_i) \} = \exp\{n\rho \mathcal{R}_{\frac{1}{1+\rho}} (\{{X}_i\})\} \label{eqn61}
\end{eqnarray}
due to costs are defined to be integers in our application scenario.  


 \subsection{Data Repair with Multiple Passes: Density and Cost Evolution}
 
In the previous subsection we have considered a static case i.e., a realization of an LDPC code ensemble i.e., a fixed bipartite graph representation. On the other hand, check and variable node degrees of a typical LDPC code ensemble is  governed by degree distributions. As can be seen in Fig. \ref{fig:sgc2}, the variable node of interest as well as its neighboring check nodes of degrees $d_{u_1},d_{u_2},\dots,d_{u_{d_v}}$, are shown. One can think of these values as realizations of the variable and check node degree distributions of LDPC codes typically expressed in polynomial forms as $\Lambda(x)=\sum_{d=1}^{D_v} \Lambda_d x^d$ and $\Phi(x)=\sum_{d=1}^{D_c} \Phi_d x^d$, respectively. Furthermore, we can define \textit{edge-perspective} degree distributions for variable and check nodes in terms of node-perspective ones as follows \cite{Richardson2001},
\begin{eqnarray}
\lambda(x)= \frac{\Lambda'(x)}{\Lambda'(1)}=\sum_{d=1}^{D_v} \lambda_d x^{d-1}, \ \ \ \phi(x)= \frac{\Phi'(x)}{\Phi'(1)} = \sum_{d=1}^{D_c} \phi_d x^{d-1}.
\end{eqnarray}
where the code rate ($r_{LDPC}$) can be described in terms of  edge-perspective degree distributions as follows
\begin{eqnarray}
r_{LDPC} = \frac{k}{n} = 1 - \frac{\int_0^1 \phi(x)dx}{\int_0^1 \lambda(x)dx} = 1 - \frac{\sum_d \phi_d/d}{\sum_d \lambda_d/d}.
\end{eqnarray}
where $\phi_d (\lambda_d)$ is the probability that when we select an edge from the underlying bipartite graph randomly, it belongs to the set of the edges of a degree-$d$ check (variable) node.

In Proposition \ref{prop2},  we have 1) conditioned on the node degrees of variable and check nodes and 2) considered only a single pass of the iterative repair strategy.  
Also, depending on the node failure patterns, it is likely that none of the check relations would be able to help with the repair process in the initial pass which would require us to decide on the successful completion of the repair process. One option is to download the missing content from the backup master and cease the repair process. The alternative option is to execute one more iteration to reduce the slave node unavailability/failure probability\footnote{For this option to be reliable, we have to have the simultaneous repair successes of the other slave nodes which were to be repaired in the previous pass and the assumption that no further node losses occur during the consecutive iterations.}. Note that in this scenario, the node repairs are decentralized and take place in the absence of node unavailability/failure information. 


Let $\{c_{u_1},\dots,c_{u_{d_v}}\}$ be the list of random variables characterizing the costs of contacting the check nodes $u_1,\dots,u_{d_v}$, and $\{c_{u_{(1)}}\leq \dots \leq c_{u_{(d_v)}}\}$ denote these random variables rearranged in
non-decreasing order of magnitude with $c_{max} = c_{u_{(d_v)}}$ representing the maximum of the cost values. Based on proposition \ref{prop2},  an optimal guessing strategy shall order the check nodes on the basis of their degrees (i.e., costs) assuming independent node failures. Accordingly, let us define 
\begin{eqnarray}
\phi^{(z)}(x) \triangleq \sum_{d=1}^{D_c} \phi_d^{(z)} x^{d-1}
\end{eqnarray}
to be the distribution of the $z$-th order statistic ($z$-th smallest) of the costs i.e., $c_{u_{(z)}} \sim \phi^{(z)}(x)$ for $1 \leq z \leq d_v$. Note that here $\phi_d^{(z)}$ refers to the probability that a randomly selected edge belongs to a degree-$d$ check node which gets selected in the $z$-th position in our guessing strategy when we order costs in non-decreasing order. Then, given that the variable node under repair has $d_v$ check options, the probability that $j$-th check node selection of the optimal guessing strategy will successfully repair the lost node in the $l$-th pass can be shown based on conditioning arguments to be\footnote{Here, we note that the repairing variable node does not download the corresponding symbols unless the check relations ensure that the repair can complete successfully.}
 \begin{align}
p_j^{(l)} (\mathbf{c}^{(j)}) = \sum_{d \in \mathbf{c}^{(j)}+1} \phi_{d|d_v}^{(j)}(1-\epsilon_l)^{d-1} \prod_{i=1}^{j-1} \left(1-\sum_{d \in \mathbf{c}^{(j)}+1} \phi_{d|d_v}^{(i)}(1-\epsilon_l)^{d-1}\right),
 \end{align}
where $\epsilon_l$ is the loss probability of a randomly chosen node at the $l$-th iteration and $\phi_{d|d_v}^{(i)}$ is the conditional probability that $i$-th check node neighbor of a variable node having degree $d_v$ has degree $d$ when neighboring check node degrees are sorted in non-decreasing order and can be expressed in a closed-form as follows,
\begin{align} 
    \phi_{d|d_v}^{(i)}&= \sum_{l_1=0}^{i-1}\sum_{l_2=0}^{i-1-l_1}\sum_{l_3=0}^{i-1-l_1-l_2}\hdots \sum_{l_{d}=0}^{i-1-l_1-l_2-\hdots-l_{d-1}}
    \prod_{t=1}^{d}\Bigg(\phi_t^{l_t} \hdots \\\nonumber & \ \ \ \ \ \phi_{d}\sum_{r_1=0}^{d_v-i}\sum_{r_2=0}^{d_v-i-r_1}\hdots\sum_{r_{D_c-d+1}=0}^{d_v-i-r_1-r_2-\hdots r_{D_c-i}}
    \prod_{y=d}^{D_c}\phi_{y}^{r_{y-d+1}}\frac{d_v!}{(\prod_{h=1}^{d-1}l_{h}!\prod_{h=2}^{D_c-i+1}r_{h}!)(l_{d-1}+r_{r_1}+1)!}\Bigg).
\end{align}
Since we conditioned on  $d_v$ check options for a recovering variable node, we realize that we will not be able to complete the repair process if none of the check relations are able to help, which happens with probability at the $l$-th iteration
\begin{align}
    1 - \sum_{j=1}^{d_v} p_j^{(l)}(\mathbf{c}^{(j)}) &= \prod_{i=1}^{d_v} \left(1-\sum_{d \in \mathbf{c}^{(j)}+1} \phi_{d|d_v}^{(i)}(1-\epsilon_l)^{d-1}\right). 
\end{align}

To remove conditioning, we sum over all possibilities of both sides and obtain
 \begin{align}
    1 - \sum_{d_v=1}^{D_v}\lambda_{d_v}\sum_{j=1}^{d_v} p_j^{(l)}(\mathbf{c}^{(j)}) &=
    \sum_{d_v=1}^{D_v} \lambda_{d_v} \prod_{i=1}^{d_v} \left(1-\sum_{d \in \mathbf{c}^{(j)}+1} \phi_{d|d_v}^{(i)}(1-\epsilon_l)^{d-1}\right) \\
    &= \sum_{d_v=1}^{D_v} \lambda_{d_v} \prod_{i=1}^{d_v} \left(1-\sum_{d \in \mathbf{c}^{(j)}+1} \phi_{d}^{(i)}(1-\epsilon_l)^{d-1}\right) \\
    & = \sum_{d_v=1}^{D_v} \lambda_{d_v}\left(1-\sum_{d \in \mathbf{c}^{(j)}+1} \phi_d(1-\epsilon_l)^{d-1}\right)^{d_v} \\
    & =  \sum_{d_v=1}^{D_v} \lambda_{d_v} (1-\phi(1-\epsilon_l))^{d_v} = \lambda(1-\phi(1-\epsilon_l))
\end{align}
which clearly does not depend on the guessing strategy since the local repair process already fails. From these equalities we observe that we readily have 
 \begin{align}
     \prod_{i=1}^{d_v} \left(1-\sum_{d \in \mathbf{c}^{(j)}+1} \phi_{d}^{(i)}(1-\epsilon_l)^{d-1}\right)\neq \left(1-\sum_{d \in \mathbf{c}^{(j)}+1} \phi_d(1-\epsilon_l)^{d-1}\right)^{d_v}=  \prod_{i=1}^{d_v} \left(1-\sum_{d \in \mathbf{c}^{(j)}+1} \phi_{d|d_v}^{(i)}(1-\epsilon_l)^{d-1}\right).
 \end{align}

The recovery failure probability of a given unavailable/failed slave node now evolves (due to the assumption of independence and averaging over the edge-perspective variable node degrees) and can be expressed as
\begin{eqnarray}
\epsilon_{l+1} = \epsilon_0 \sum_{d=1}^{D_v} \lambda_d (1-\phi(1-\epsilon_l))^{d-1} = \epsilon_0 \lambda(1-\phi(1-\epsilon_l))
\label{eqn13}
\end{eqnarray}
which brings us to the standard density evolution formula for erasure channels. 

\begin{figure}[t!]
\centering
  \includegraphics[width=0.65\linewidth]{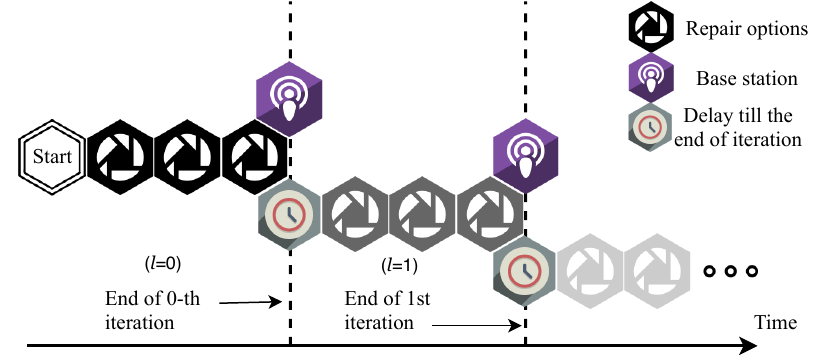}
  \caption{An example of multi-iteration repair process as a function of time for an LDPC code symbol that has a node degree of $3$. At the end of each iteration, the process can either choose to contact the back-up node or another round of iteration to complete the repair. Repair options colored as black hexagons get less intense as we have better chances of repair completion.}
  \label{fig:ldpcnode3}
  \vspace{-4mm}
\end{figure}

\subsection{A deferred master-node communication protocol} 
Let us assume there is a backup master (a base station) to help with the repair process within the cell. In this case, we assume that we can directly download the contents from this backup at the expense of a fixed cost $C_\pi (> 1)$ per symbol download. In this simple protocol, we aim at maximizing (minimizing) the use of local nodes (master nodes) in the repair process. More specifically, we order the check relations based on the degrees i.e., bandwidth cost of repair, and utilize multiple iterations to ensure the regeneration. Since the node availability information is missing at the time of the repair, we confirm whether all connections of the neighbors of the first check node can successfully be established. If at least one of the variable nodes can not be reached, we next check the availability of neighbors of the next check node and so on. The lost node is repaired (actual download happens) using neighbors of a check node whose all neighboring variable nodes are intact and reachable.   However, if this attempt is not successful at the current iteration, we have to decide whether to reach the backup for the completion of the repair process or take another round of iteration, unless a predetermined maximum number of iterations is reached.  An example is shown in Fig. \ref{fig:ldpcnode3} for $d_v=3$. As can be seen, at the end of each iteration, a decision is made whether to complete the repair process with the help of a backup or continue with another round of iteration. Since we may have downloaded the contents of the repaired node directly from the backup, we first check if the backup node is to be contacted at the end of contacting all local nodes for a given iteration round. If so, we allow another round of iteration and if not, we cease iterations and complete the repair with the help of the backup node. Therefore, with this deferred master-node communication protocol, to be able to ease the analysis, the backup node is allowed to be contacted only at the end of each iteration round and the advantage of conducting data regeneration using local nodes is maximized.

Accordingly, to contact the backup at the beginning of the $l$-th iteration for $l \geq 1$, we need to make sure that the backup node would not be contacted last within the same iteration i.e., the cost of contacting back-up is not too costly compared to local downloads. Recall from Theorem \ref{prop2} that the back-up is not considered as a last resort only when the condition $C_{\pi} < c_{max} \left((1-q\right)^{-c_{max}}-1)$ is met. In other words, when the repair process comprises several iterations, to contact the back-up at the beginning of the $l$-th iteration for $l \geq 0$, we need to make sure that it would not be contacted at the end of the current iteration.  For the LDPC code ensembles, the above condition happens with a non-zero probability since $c_{max}$ is a random variable. For a given variable node degree $d$, we contact the backup node at the beginning of the $l$-th iteration with conditional probability $\tau_{l|d}$ given by
\begin{eqnarray}
\tau_{l|d} = Pr\left( (C_{\pi} < c_{max} \left((1-\epsilon_{l}\right)^{-c_{max}}-1))  \wedge (C_{\pi} \geq c_{max} \left((1-\epsilon_{t-1}\right)^{-c_{max}}-1) ~\forall ~t \in [l]) \right)
\end{eqnarray}
which can be simplified due to the monotonicity of $c_{max} \left((1-\epsilon_{l}\right)^{-c_{max}}-1))$ as (see also Remark \ref{remark41})
\begin{eqnarray}
\tau_{l|d} = Pr\left( (C_{\pi} < c_{max} \left((1-\epsilon_{l}\right)^{-c_{max}}-1))  \wedge (C_{\pi} \geq c_{max} \left((1-\epsilon_{l-1}\right)^{-c_{max}}-1) ) \right).
\end{eqnarray}
Note that since successful repair is guaranteed when the back-up node is involved, based on the  above formulation,  the evolution formula in \eqref{eqn13} can be rewritten as (again with the And-Or tree assumption \cite{lubyandortree1998} in the decoding path)
\begin{eqnarray} \label{eq:newDensityEvolution}
\epsilon_{l+1} &=& \epsilon_0  \sum_{d=1}^{D_v} \lambda_d (1-\tau_{l|d})(1-\phi(1-\epsilon_l))^{d-1}.
\end{eqnarray}

On the other hand, we notice that the conditional probability that we contact the back-up at the end of the $l$-th iteration i.e., $1-\rho_{l|d}$ can be approximated for small $\epsilon_0 c_{max}$ (i.e., $\epsilon_0 D_c$) as
\begin{eqnarray}
1 - \tau_{l|d} &=& Pr\left(c_{max} \left((1-\epsilon_{l+1}\right)^{-c_{max}}-1) \leq C_{\pi} \right) \\
&=& Pr\left( c^2_{max} \epsilon_{l+1} \leq c_{max} ( \left (1-\epsilon_{l+1}\right)^{-c_{max}}-1) \leq C_{\pi} \right) \\
&\approx& Pr\left(c_{max} \leq \sqrt{C_{\pi}/\epsilon_{l+1}} \right) \label{eqn79}\\
&=& \left( \sum_{i=1}^{\min \{ D_c, \lfloor \sqrt{C_{\pi}/\epsilon_{l+1}}\rfloor \} }  \Phi_i \right)^{d} \label{eqn65}
\end{eqnarray}
where \eqref{eqn65} follows due to independence of $c_i$'s.

\begin{remark} \label{remark41}
Note that in this setting, as long as $\epsilon_l \rightarrow 0$ as $l$ tends large, we have $\epsilon_{l+1} \leq \epsilon_l$  which leads to the relationship $\tau_{l|d} \geq \tau_{l+1|d}$ i.e., as the iterations get deeper, it becomes less likely to contact the back-up node for the repair due to reduced loss probabilities of the neighboring nodes. 
\end{remark}

\subsection{Decoding Threshold with Back-up} For ease of analysis, let us assume small $\epsilon_0 D_c$ and not put any limit on the number of iterations with a predefined threshold. In this case, we notice that if $\sqrt{C_\pi/\epsilon_{0}} \geq D_c$, then this would result in standard density evolution and the decoding threshold, in that case, would be defined to be
\begin{eqnarray}
\epsilon_0^* = \sup \{\epsilon_0 \lambda(1-\phi(1-x))<x, \forall x, x \in(0,\epsilon_0]\},
\end{eqnarray} 
i.e., the maximal erasure probability below which error-free repair is possible through solely using iterations/local nodes. On the other hand if $\sqrt{C_\pi/\epsilon_{0}} < D_c$, suppose in one of the iterations of the BP (say $l^*$-th iteration), we have $\epsilon_0 \geq  \epsilon_{l^*-1} \geq \epsilon_{l^*}$ such that $\lceil \sqrt{C_\pi/\epsilon_{l^*-1}} \rceil \leq D_c \leq \lceil \sqrt{C_\pi/\epsilon_{l^*}} \rceil$, then for all $l \leq l^*-2$ we would have $\epsilon_{l+1}$ to be the solution to the following equation
\begin{align}
\epsilon_{l+1} = \epsilon_0  \sum_{d=1}^{D_v} \lambda_d \left( \sum_{i=1}^{\lfloor \sqrt{C_{\pi}/\epsilon_{l+1}}\rfloor }  \Phi_i \right)^{d}  (1-\phi(1-\epsilon_l))^{d-1} 
\end{align}
and finally for $l \geq l^* - 1$, $\epsilon_{l+1}$ is given by the standard density evolution formula. Accordingly, the decoding threshold in that case is given by
\begin{align}
\epsilon_0^{\dagger} (C_\pi) &= \sup \left\{\epsilon_0  \sum_{d=1}^{D_v} \lambda_d \left( \sum_{i=1}^{ \min\{D_c,\lfloor \sqrt{C_{\pi}/\epsilon_{l+1}}\rfloor\} }  \Phi_i \right)^{d}  (1-\phi(1-x))^{d-1} <x, \forall x, x \in(0,\epsilon_0]\right\} \\
&= \inf \left\{ \frac{x}{\sum_{d=1}^{D_v} \lambda_d \left( 1 - \sum_{i=1}^{ \min\{D_c,\lfloor \sqrt{C_{\pi}/\epsilon_{l+1}}\rfloor\} }  \Phi_i \right)^{d}  (1-\phi(1-x))^{d-1}}, \forall x, x\in (0,1) \right\}
\end{align} 

Here we immediately realize the relationship $\epsilon_0^{*} \leq \epsilon_0^{\dagger} (C_\pi) \leq 1$ i.e., the decoding threshold can be improved with the help of a master back-up node in the context of data reconstruction.  

\subsection{Numerical {Demonstration}}

We consider an irregular LDPC code that performs provably close to the optimal (achieving minimum gap to the channel capacity) under BEC \cite{Amraoui2007}. The edge-perspective degree distributions of this code ensemble are given by
\begin{eqnarray}
\phi(x) &=& 0.608291x^5 + 0.391709x^6 \label{LDPCCodeCheckDist}\\ 
\lambda(x) &=& 0.205031x + 0.455716x^2 + 0.193248x^{13} + 0.146004x^{14}  \label{LDPCCodeVarDist}  
\end{eqnarray}
from which the rate of the code can be calculated to be $0.4339$ with $D_c=7$ and $D_v=15$. 
The results of our cost evolution process are presented
in which $\tau_{l|d}$ is estimated numerically based on
\begin{eqnarray}
\tau_{l|d} \approx Pr\left(C_{\pi} < c_{max} \left((1-\epsilon_{l}\right)^{-c_{max}}-1)\right) \times Pr\left(C_{\pi} \geq c_{max} \left((1-\epsilon_{l-1}\right)^{-c_{max}}-1)\right)
\end{eqnarray}    
rather than the approximation given by the Eqn. \eqref{eqn79} since $\max \{ \epsilon_0 D_c \} = 3.037$ is not small enough.
\begin{figure}[t!]
\centering
  \includegraphics[width=0.75\linewidth]{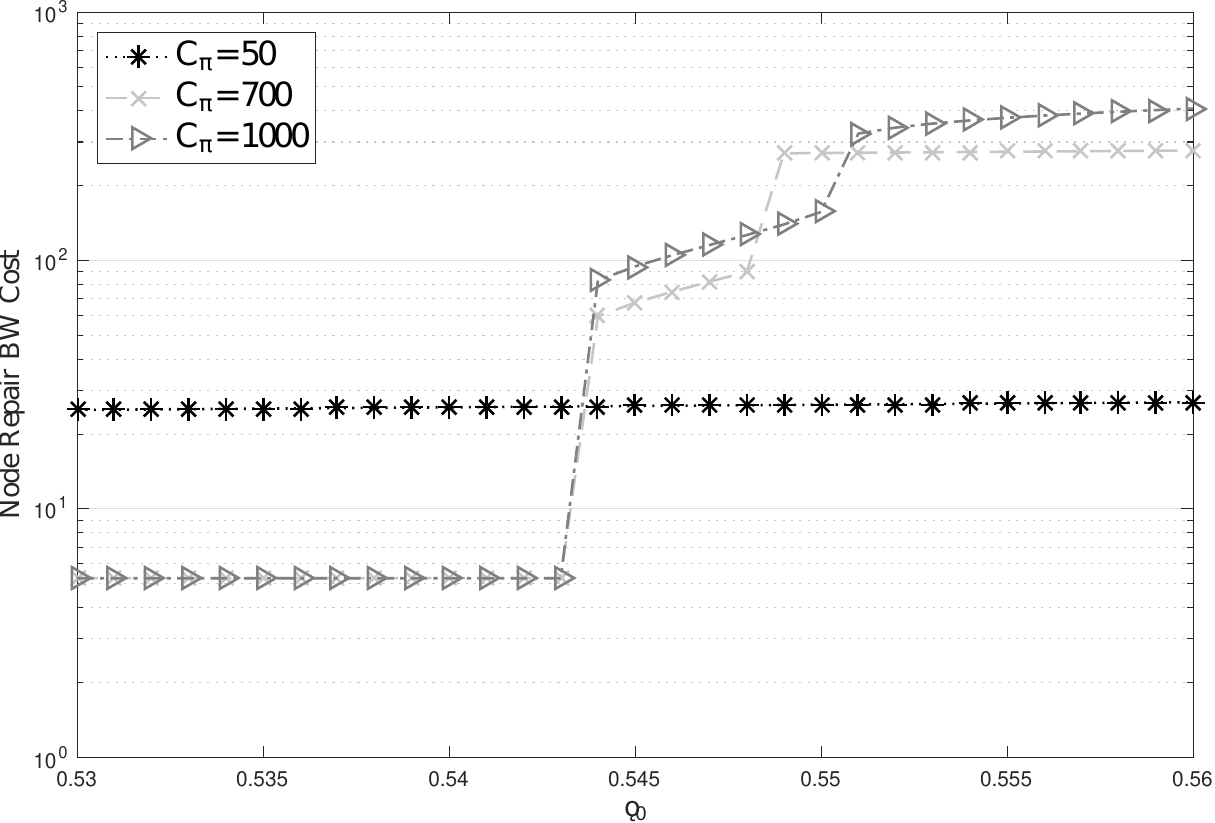}
  \caption{Cost of a symbol repair with a back-up Master}
  \label{fig:ldpcCosts}
\end{figure}

The cost of a symbol repair using the strategy with back-up master node having $C_\pi \in \{50,700,1000\}$, respectively, is provided in Fig. \ref{fig:ldpcCosts}. 
When $C_\pi=50$, the the use of the base station has commenced in early stages of the node repair process (iterations). In other words, in case of  $C_\pi=50$, the master node is used as the last resort without further iterations are performed for $\epsilon \in [0.53,0.56]$. When $C_\pi=700$ and $\epsilon_0\geq0.544$, the use of backup master is preferred, which in turn leads to higher node repair cost.  Moreover, when $C_\pi=1000$ and $\epsilon_0$ gets the value of $0.551$, the value of $\epsilon_{max}$ is increased substantially. For the value of $C_\pi=1000$, the use of BS before the last iteration is started to occur when $\epsilon_0\geq0.575$.

In Table \ref{tab:LDPCCosts}, the results of actual optimal LDPC repair cost as well as upper and lower bounds obtained through numerical evaluations of Thm. \ref{theo:boztasUpperBoundupdatedtheo3_4}, and Eq. (\ref{sasonlb}) are presented for all combinations of $\rho \in \{1,2,3\}$ and $\epsilon_0 \in \{0.01,0.05,0.1\}$. For this data repair scenario, the base station cost is  assumed to be $C_\pi=1000$. Based on our numerical results for $\rho=1,~\rho=2,~\rho=3$, the evaluation of Thm. \ref{theo:boztasUpperBoundupdatedtheo3_4} gives   $64\%, ~22\%,~17\%$ higher results than the actual results on average, whereas Eq. (\ref{sasonlb}) provides $24\%,~12\%,~11\%$ lower values than the actual on average, respectively.

\begin{table}[t!]
\caption{{Numerical computation} of $\frac{1}{n} \ln ( \mathbb{E}[C_{\mathcal{G}^*}(X)^\rho])$ }
    \label{tab:LDPCCosts}
    \centering
    \begin{tabular}{|c|c|c|c|c|}
    \hline
         $\boldsymbol{\rho}$ \textbf{values}&$\boldsymbol{\epsilon_0}$ \textbf{values} & \textbf{0.01}&\textbf{0.05} &\textbf{0.1}  \\
         \cline{1-5}
         \parbox[t]{25mm}
         
         {\multirow{3}{*}{~~~~~
         \rotatebox[origin=c]{0}{$~~\boldsymbol{ \rho=1 }~~~~~~$}}}& \textbf{$\boldsymbol{Thm. \ref{theo:boztasUpperBoundupdatedtheo3_4}}$}  & 5.40& 5.184& 5.342\\ \cline{2-5}
        & \textbf{$\boldsymbol{OPT}$} &2.099&3.996&4.967
         \\\cline{2-5}
        & \textbf{$\boldsymbol{Eq. (\ref{sasonlb})}$} &1.040&3.443&4.502  \\\hline

         \parbox[t]{25mm}
         
         {\multirow{3}{*}{~~~~~
         \rotatebox[origin=c]{0}{$~~\boldsymbol{ \rho=2 }~~~~~~$}}}& \textbf{$\boldsymbol{Thm. \ref{theo:boztasUpperBoundupdatedtheo3_4} }$} &6.216 &6.047& 6.105 \\ \cline{2-5}
        & \textbf{$\boldsymbol{OPT}$} & 4.048&5.431&5.926
         \\\cline{2-5}
        & \textbf{$\boldsymbol{Eq. (\ref{sasonlb}).}$} &3.253&4.882&5.465 \\ \hline
        
           \parbox[t]{25mm}
         
         {\multirow{3}{*}{~~~~~
         \rotatebox[origin=c]{0}{$~~\boldsymbol{ \rho=3 }~~~~~~$}}}& \textbf{$\boldsymbol{Thm. \ref{theo:boztasUpperBoundupdatedtheo3_4} }$}  &6.260 &6.152&6.222 \\ \cline{2-5}
        & \textbf{$\boldsymbol{OPT}$} &4.461& 5.652&6.075
         \\\cline{2-5}
        & \textbf{$\boldsymbol{Eq. (\ref{sasonlb})}$} &3.694&5.105&5.616 \\\hline
        
        
    \end{tabular}
    
\end{table}


 

\section{Conclusions and Future work}

In this work, the general notion of guessing cost is introduced and an optimal strategy is provided for guessing a random variable defined on a finite set with each choice may be associated with a positive finite cost. Upper and lower bounds on the moments of guessing cost are derived and expressed in terms of the Rényi's entropy and entropy rate. We {have established connections} with the guesswork through {introducing} novel random variables. Thanks to this connection, previous works on the improvements of upper/lower bounds for the guesswork become readily applicable. Accordingly, we provided improved bounds on the moments of guessing cost without lengthy proofs.  Finally, we established the guessing cost exponent on the moments of the optimal guessing by considering a sequence of random variables. These bounds are shown to serve quite useful for bounding the overall repair latency cost (data repair complexity) for distributed data storage systems in which sparse graph codes may be utilized.  We have assumed a simple protocol to derive initial results and demonstrated the usefulness of the previously derived bounds. {It's important to highlight that in the design of a distributed storage protocol, there may be value in giving up the prediction of the next value based on conditions like the total accumulated cost.} Characterization of the guessing cost, in that case, would have to be expressed in terms of smooth Rényi's entropy. Recent studies such as \cite{kuzuoka2019} considered similar constraints for the guesswork within the context of source coding. Such scenarios would be considered in our future work to be able to improve protocol design towards better system maintenance in presence/assistance of an external back-up/base station master.

\appendices
\section{Proof of Proposition \ref{prop0}}

Let us start with $\rho=1$ i.e., mean guessing cost given in Remark \ref{prop1}. Consider swapping the $i$-th and $(i+1)$-th guessed values. Let $\mathcal{G}_{i,i+1}$ be the original guessing strategy and $\mathcal{G}_{i+1,i}$ be the swapped version.  Then it is straightforward to show that the difference is 
\begin{align}
    \mathbb{E}[C_{\mathcal{G}_{i,i+1}}(x)] - \mathbb{E}[C_{\mathcal{G}_{i+1,i}}(x)] &= c_i(1-g_{i-1}) \\ 
    &+ c_{i+1}(1-g_{i-1}-p_i) \\ 
    &- c_{i+1}(1-g_{i-1}) \\
    &- c_i(1-g_{i-1}-p_{i+1}) \\
    &= c_ip_{i+1} - c_{i+1}p_i
\end{align}
which implies that if $c_ip_{i+1} > c_{i+1}p_i $, then we swap $i$-th and $(i+1)$-th guessed values in order to reduce the average  guessing cost, otherwise no swapping is performed. Since each swapping leads to lower cost,  for any $i, j \in\{1,\dots,M\}$ with $i \leq j$,  the optimal guessing strategy $\mathcal{G}^*$ would satisfy the following series of inequalities
 \[
    \begin{array}{ll}
        c_ip_{i+1} &\leq c_{i+1}p_i \\
        c_{i+1}p_{i+2} &\leq c_{i+2}p_{i+1} \\ &\vdots \\
        c_{j-1}p_{j} &\leq c_{j}p_{j-1}
        \end{array} \Vast\} \Rightarrow c_i p_j
        \prod_{k=i+1}^{j-1} c_kp_k \leq c_j p_i
        \prod_{k=i+1}^{j-1} c_kp_k 
  \]
where continuing deriving the inequality and multiplying left-hand and right-hand consecutive terms individually would give us the desired result since all $p_i$s and $c_i$s are non-negative.

Now, let us consider the general case i.e., for any real $\rho>0$, we have
\begin{align}
\mathbb{E}[C_\mathcal{G}(X)^\rho] &= \sum_{i=1}^M \left(\sum_{j=1}^{i} c_j\right)^\rho p_i = \sum_{i=1}^M ||\mathbf{c}^{(i)}||_{1}^\rho p_i  
\end{align}

In general, we would be looking for a condition that would ensure the following for indices $i,j$ satisfying $i \leq j$
\begin{eqnarray}
 \mathbb{E}[C_{\mathcal{G}_{i,j}}(x)^\rho] - \mathbb{E}[C_{\mathcal{G}_{j,i}}(x)^\rho] \leq 0
\end{eqnarray}
which would mean that swaping these indices do not improve the moments of guessing cost. This condition can be shown to imply for any $\rho \in (0,+\infty)$ through some algebra that
\begin{align}
    \left[ ||\mathbf{c}^{(i)}||_{1}^\rho - ||\mathbf{c}^{(j)}||_{1}^\rho \right] p_i + \left[ ||\mathbf{c}^{(j)}||_{1}^\rho - (||\mathbf{c}^{(i)}||_{1}-c_i+c_j)^\rho \right] p_j \leq  \sum_{l=i+1}^{j-1}  \left[  (||\mathbf{c}^{(l)}||_{1}-c_i+c_j)^\rho - ||\mathbf{c}^{(l)}||_{1}^\rho \right] p_l \label{cond4generalrho}
\end{align}

If we had consider $i$-th and $(i+1)$th indices instead, then the condition we have derived for $\rho=1$ i.e., $c_ip_{i+1} \leq c_{i+1}p_i$ would have been extended to any $\rho > 0$ for all indices between $i$ and $j$. We can finally arrive at the following set of inequalities
\begin{eqnarray}
\left[||\mathbf{c}^{(i+1)}||_{1}^\rho - (||\mathbf{c}^{(i+1)}||_{1}-c_i)^\rho\right] p_{i+1} &\leq&  \left[||\mathbf{c}^{(i+1)}||_{1}^\rho - ||\mathbf{c}^{(i)}||_{1}^\rho\right]p_i \\
\left[||\mathbf{c}^{(i+2)}||_{1}^\rho - (||\mathbf{c}^{(i+2)}||_{1}-c_{i+1})^\rho\right] p_{i+2} &\leq&  \left[||\mathbf{c}^{(i+2)}||_{1}^\rho - ||\mathbf{c}^{(i+1)}||_{1}^\rho\right]p_{i+1} \\
\vdots \\
\left[||\mathbf{c}^{(j)}||_{1}^\rho - (||\mathbf{c}^{(j)}||_{1}-c_{j-1})^\rho\right] p_{j} &\leq&  \left[||\mathbf{c}^{(j)}||_{1}^\rho - ||\mathbf{c}^{(j-1)}||_{1}^\rho\right]p_{j-1}
\end{eqnarray}
Note that due to non-negativity of costs, for any $l \in \{2,\dots,j-i\}$ and $\rho \geq 1$, we always have 
\begin{align}
    ||\mathbf{c}^{(i+l)}||_{1}^\rho - (||\mathbf{c}^{(i+l)}||_{1}-c_{i+l-1})^\rho  \geq ||\mathbf{c}^{(i+l-1)}||_{1}^\rho - ||\mathbf{c}^{(i+l-2)}||_{1}^\rho
\end{align}
which can be rewritten as 
\begin{align} (||\mathbf{c}^{(i+l-1)}||_{1} + c_{i+l})^\rho- (||\mathbf{c}^{(i+l-2)}||_{1}+ c_{i+l})^\rho  \geq ||\mathbf{c}^{(i+l-1)}||_{1}^\rho - ||\mathbf{c}^{(i+l-2)}||_{1}^\rho.
\end{align}
Then in order to satisfy all the inequalities above, we must have
\begin{align}
    \left[||\mathbf{c}^{(i+1)}||_{1}^\rho - (||\mathbf{c}^{(i+1)}||_{1}-c_i)^\rho\right]p_j \leq \left[||\mathbf{c}^{(j)}||_{1}^\rho - ||\mathbf{c}^{(j-1)}||_{1}^\rho\right] p_i
\end{align}
which is a simplified condition as compared to the condition in \eqref{cond4generalrho} for $\rho \in (1,+\infty)$.


{
\section{Proof of Proposition \ref{prop00}}}

{Let us assume $c_x \in \mathbb{Z}^+$. Our objective is to show that $\overline{\mathcal{G}}:\overline{\mathcal{X}}\rightarrow \left[ \sum_x c_x  \right]$ is a bijection. Let us begin with one-to-one property and with a trivial case. Suppose that $j=i+1 \not= i$ for $i < \sum_x c_x$, and assume that the following relation is true. 
\begin{eqnarray}
    \overline{\mathcal{G}}(\overline{X} = \overline{x}_i)  = \overline{\mathcal{G}}(\overline{X} = \overline{x}_j) 
    = \overline{\mathcal{G}}(\overline{X} = \overline{x}_{i+1}) \label{fundeq}
\end{eqnarray}
which obviously violates the 1-to-1 property. For a given index {$i \in {[\sum_x  c_x ]}$}, due its construction in the manuscript, there exists a positive integer $k^{(i)}\leq M$ satisfying
\begin{align}
    \sum_{x \in \mathcal{X}^{k^{(i)}-1}}  c_x  < i \leq  \sum_{x \in \mathcal{X}^{k^{(i)}}}  c_x. \label{ineq9}
\end{align}}
{where $\mathcal{X}^i = \{x_1,\dots,x_i\}$ with $\mathcal{X}^0 = \varnothing$. If $k^{(i)} = k^{(i+1)}$, then it is easy to see that the equality \eqref{fundeq} 
would be impossible to hold since $i \not= i + 1$ in the definition of $ \overline{\mathcal{G}}$. If $k^{(i)} \not= k^{(i+1)}$ we notice that we can use the expression given for $ \overline{\mathcal{G}}$ and rewrite  both sides of the equation \eqref{fundeq} as}
{
\begin{eqnarray}
    \overline{\mathcal{G}}(\overline{X} = \overline{x}_i) &=& \sum_{x: \mathcal{G}(x) < \mathcal{G}\left(x_{k^{(i)}}\right)}  c_x  - \sum_{x \in \mathcal{X}^{k^{(i)}-1}}  c_{x} + i \\
    &=& \sum_{x: \mathcal{G}(x) < \mathcal{G}\left(x_{k^{(i)}}\right)}  c_x  - \sum_{x \in \mathcal{X}^{k^{(i)}}}  c_{x}  + i +  c_{x_{k^{(i)}}}   \label{eqn8} \\
    \overline{\mathcal{G}}(\overline{X} = \overline{x}_{i+1]}) &=& \sum_{x: \mathcal{G}(x) < \mathcal{G}\left(x_{k^{(i+1)}}\right)}  c_x - \sum_{x \in \mathcal{X}^{k^{(i+1)}-1}}  c_{x}  + i + 1 \\
    &=& \sum_{x: \mathcal{G}(x) < \mathcal{G}\left(x_{k^{(i+1)}}\right)}  c_x - \sum_{x \in \mathcal{X}^{k^{(i)}}} c_{x}  + i + 1. \label{eqn10}
\end{eqnarray}}
{where we used the fact that $\mathcal{X}^{k^{(i+1)}} = \mathcal{X}^{k^{(i)}+1}$ due to the way $\mathcal{X}^{i}$ is defined. Note that since $k^{(i)} \not= k^{(i+1)}$ is assumed, the first terms in the expressions of \eqref{eqn8} and \eqref{eqn10} cannot be equal to satisfy the equation \eqref{fundeq}. Therefore, we have two possible cases.}
{
\begin{itemize}
    \item Case $\sum_{x: \mathcal{G}(x) < \mathcal{G}\left(x_{k^{(i)}}\right)}  c_x  > \sum_{x: \mathcal{G}(x) < \mathcal{G}\left(x_{k^{(i+1)}}\right)}  c_x $: In this case however, it can be clearly seen that we must have  $\overline{\mathcal{G}}(\overline{X} = \overline{x}_i)  > \overline{\mathcal{G}}(\overline{X} = \overline{x}_{i+1})$ since $ c_{x_{k^{(i)}}}  \geq 1$. 
    \item Case $\sum_{x: \mathcal{G}(x) < \mathcal{G}\left(x_{k^{(i)}}\right)}  c_x  < \sum_{x: \mathcal{G}(x) < \mathcal{G}\left(x_{k^{(i+1)}}\right)} c_x $: In this case we must have 
    \begin{eqnarray}
        \sum_{x: \mathcal{G}(x) < \mathcal{G}\left(x_{k^{(i+1)}}\right)}  c_x  - \sum_{x: \mathcal{G}(x) < \mathcal{G}\left(x_{k^{(i)}}\right)}  c_x  \geq  c_{x_{k^{(i)}}} 
    \end{eqnarray}
    due to the ordering of costs. This result implies that $\overline{\mathcal{G}}(\overline{X} = \overline{x}_i)  < \overline{\mathcal{G}}(\overline{X} = \overline{x}_{i+1})$, which is necessarily a strict inequality due to the assumption $j = i + 1 \not= i$. 
\end{itemize}}

{As a result, our initial assumption that $    \overline{\mathcal{G}}(\overline{X} = \overline{x}_i)  = \overline{\mathcal{G}}(\overline{X} = \overline{x}_{i+1})$ cannot be true. Using this observation, we can extend our arguments to any $(i, j)$ pair with $i \not= j$. WOLOG, assume that $j>i$, by considering the pairs $(i,i+1),(i+1,i+2),\dots,(j-1,j)$ in this particular order, it is not hard to show $    \overline{\mathcal{G}}(\overline{X} = \overline{x}_i)  \not= \overline{\mathcal{G}}(\overline{X} = \overline{x}_{j})$ i.e., $\overline{\mathcal{G}}$ is one-to-one. On the other hand,}
{we also notice that 
\begin{eqnarray}
\overline{\mathcal{G}}(\overline{X} = \overline{x}_i) &=& \sum_{x: \mathcal{G}(x) < \mathcal{G}\left(x_{k^{(i)}}\right)}  c_x  + \underbrace{i - \sum_{x \in \mathcal{X}^{k^{(i)}-1}} c_{x} }_\text{$>0$ due to Eqn. \eqref{ineq9}} \\
&>& 0 \label{eqn00}
\end{eqnarray}
due to non-negativity of costs. Hence, the minimum integer the strategy could map to is 1. In addition, 
\begin{eqnarray}
    \overline{\mathcal{G}}(\overline{X} = \overline{x}_i) &=& 
 \sum_{x: \mathcal{G}(x) < \mathcal{G}\left(x_{k^{(i)}}\right)}  c_x  + \underbrace{i - \sum_{x \in \mathcal{X}^{k^{(i)}}}  c_{x} }_\text{$\leq 0$ due to Eqn. \eqref{ineq9}}  +  c_{x_{k^{(i)}}}  \\
 &\leq&  \sum_{x: \mathcal{G}(x) \leq \mathcal{G}\left(x_{k^{(i)}}\right)}  c_x  \leq \sum_{x}  c_x 
\end{eqnarray}
which, together with \eqref{eqn00}, implies $1 \leq \overline{\mathcal{G}}(\overline{X} = \overline{x}_i) \leq \sum_{x}  c_x $. As a result, we can induct that $\overline{\mathcal{G}}$ must be a bijection and hence a valid strategy/mapping.}  

\section{Proof of Theorem \ref{thm301}}
Before giving the formal proof let us state a well known lemma. 
\begin{lemma}[H\"{o}lder's inequality] 
Let $a_i$ and $b_i$ for $(i = 1, . . . , n)$ be positive real sequences. If $q > 1$
and $1/q + 1/r = 1$, then
\begin{eqnarray}
\left(\sum_{i=1}^n a_i^q\right)^{1/q}\left(\sum_{i=1}^n b_i^r\right)^{1/r} \geq \sum_{i=1}^{n} a_ib_i
\end{eqnarray}
\vspace{0.5mm}
\end{lemma}

Let $a_i$ be a positive real number for all $i$, $M$ be a natural number, and $\gamma$ be the harmonic mean of $\{a_1,\dots,a_n\}$, then we have
\begin{align}
    \sum_{i=1}^M \frac{1}{1+a_i} \leq  \frac{M}{1+\gamma} \label{eqn111}
\end{align}
which can easily be proved using Radon's inequality \cite{lai2016}. Now, let us express the lower bound of the moments of the  guessing cost as follows,
{
\begin{align}
\mathbb{E}[C_\mathcal{G}(X)^\rho] \geq \mathbb{E}[C_\mathcal{G^*}(X)^\rho] 
\geq \left[\sum_{i=1}^M \frac{1}{\sum_{j=1}^i c_j^*}\right]^{-\rho} \left[\sum_{i=1}^M P_X({\mathcal{G}^{-1}(i)})^\frac{1}{1+\rho}\right]^{1+\rho} \label{eqn133}
\end{align}
}which easily follows from a direct application of H\"{o}lder's inequality, where $\{c_j^*\}$ are the optimal ordering of cost values. To see this, let us set $r=1+ \rho$, $q=(1+\rho)/\rho$ in Hölder's inequality so that $1/q+1/r =1$ is satisfied for $\rho > 0$. We also let
\begin{align}
    a_i = \left[\sum_{j=1}^i c_{\mathcal{G}^{-1}(j)}\right]^{-\rho/(1+\rho)} \mathrm{and} \ \ \ \ 
    b_i = \left[\sum_{j=1}^i c_{\mathcal{G}^{-1}(j)}\right]^{\rho/(1+\rho)}P_X({\mathcal{G}^{-1}(i)})^{1/(1+\rho)}.
\end{align}
Now, using H\"{o}lder's inequality, it would be easy to obtain
{
\begin{align}
 \left[ \sum_{i=1}^M \frac{1}{\sum_{j=1}^ic_{\mathcal{G}^{-1}(j)}}\right]^{\rho/(1+\rho)} \left(\mathbb{E}[C_\mathcal{G}(X)^\rho]\right)^{1/(1+\rho)} \geq \sum_{i=1}^M  P_X({\mathcal{G}^{-1}(i)})^{1/(1+\rho)}
\end{align}
}from which inequality \eqref{eqn133} follows for the optimal strategy $\mathcal{G^*}$. Now, considering the ordering of costs that minimizes the right hand side, we shall have,
{
\begin{align}
\mathbb{E}[C_\mathcal{G}(X)^\rho]  &\geq  \left[\sum_{i=1}^M \frac{1}{\sum_{j=1}^i c^*_j}\right]^{-\rho} \left[\sum_{i=1}^M P_X({\mathcal{G}^{-1}(i)})^\frac{1}{1+\rho}\right]^{1+\rho}  \nonumber \\
&\geq \left(\frac{M}{1+\gamma^*}\right)^{-\rho} \left[\sum_{i=1}^M P_X(x_i)^\frac{1}{1+\rho}\right]^{1+\rho} \label{eqn14} = \left(\frac{M}{1+\gamma^*}\right)^{-\rho} \exp\left\{\rho H_{\frac{1}{1+\rho}}(X) \right\} 
\end{align} 
}where $\gamma^*$ is the harmonic mean of $\{\sum_j^i c^*_j-1\}$'s for $i=\{1,2,\dots,M\}$ and $H_{\alpha}(X)$ is Rényi's entropy of order $\alpha$ ($\alpha>0, \alpha \not=1$) for random variable $X$ defined as,
\begin{eqnarray}
H_\alpha(X) = \frac{\alpha}{1-\alpha}\ln \left[\sum_x P_X(x)^\alpha \right]^{1/\alpha}
\end{eqnarray}

Note that inequality \eqref{eqn14} followed from the inequality \eqref{eqn63}.

\section{Proof of Theorem \ref{thm31}}

Let us first observe that with the optimal guessing strategy $\mathcal{G}^*$ that minimizes the expected guessing cost $x$, 
\begin{eqnarray}
C_{\mathcal{G}^*}(x) &=& \sum_{x^\prime:C_{\mathcal{G}^*}(x^\prime) \leq C_{\mathcal{G}^*}(x)} \sum_{x^{\prime \prime}}^{c_{x^\prime}} 1\\
&\leq& \sum_{x^\prime:C_{\mathcal{G}^*}(x^\prime) \leq C_{\mathcal{G}^*}(x)} \sum_{x^{\prime \prime}}^{c_{x^\prime}} \left(\frac{c_xP_X(x^\prime)}{c_{x^\prime}P_X(x)} \right)^{\frac{1}{1+\rho}} \label{eqn20} \\
&=& \sum_{x^\prime:C_{\mathcal{G}^*}(x^\prime) \leq C_{\mathcal{G}^*}(x)} 
c_{x^\prime}^{\frac{\rho }{1+\rho}}\left(\frac{c_xP_X(x^\prime)}{P_X(x)} \right)^{\frac{1}{1+\rho}} \\  
&\leq  & \sum_{x^\prime} 
c_{x^\prime}^{\frac{\rho}{1+\rho}}\left(\frac{c_xP_X(x^\prime)}{P_X(x)} \right)^{\frac{1}{1+\rho}} \label{eqn222}
\end{eqnarray}
where the inequality \eqref{eqn20} follows from the necessary condition of Proposition \ref{prop0} with $\rho=1$ i.e., $c_{x^\prime} P_X(x) \leq c_x P_X(x^{\prime})$ for all $\{x^{\prime}:C_{\mathcal{G}^*}(x^\prime) \leq C_{\mathcal{G}^*}(x)\}$ that needs to hold for the optimal guessing strategy $\mathcal{G}^*$. Also, although the exponent $1/(1+\rho)$ decreases the value, it is still greater than 1 due to  ${\frac{c_xP_X(x^\prime)}{c_{x^\prime}P_X(x)}} \geq 1$. Using the inequality given in \eqref{eqn222} in equation \eqref{orgeqn}, we get
\begin{align}
\mathbb{E}[C_\mathcal{G^{*}}(X)^\rho] & = \sum_x P_X(x) C_{\mathcal{G}^*}(x)^\rho \\
& \leq \sum_x P_X(x) \left( \sum_{x^{\prime}} c_{x^\prime}^{\frac{\rho}{1+\rho}} \left(\frac{c_xP_X(x^\prime)}{P_X(x)}\right)^{\frac{1}{1+\rho}}
\right)^\rho = \left[\sum_x c_x^{\frac{\rho}{1+\rho}} P_X(x)^{\frac{1}{1+\rho}}\right]^{1+\rho} \\
& = \left[\sum_x c_x (P_X(x)/c_x)^\frac{1}{1+\rho}\right]^{1+\rho} \label{eqnfinal}
\end{align}

Note that for a given $i,j$ satisfying $i \leq j$, the condition $c_ip_j \leq c_jp_i$ (i.e., $c_{x^\prime} P_X(x) \leq c_x P_X(x^{\prime})$) does not necessarily imply  the condition given in \eqref{necessarycon} for any $\rho>0$. However, we observe that the general necessary condition of Proposition \ref{prop0} is more strict in the sense that the strategy that is satisfying $c_ip_j \leq c_jp_i$ (condition in \eqref{necessarycon} with $\rho=1$) would be an upper bound on the moments of guessing cost using the optimal guessing strategy. For instance, if for all $i \leq j$,
\begin{eqnarray}
    \left[ ||\mathbf{c}^{(i)}||_{1}^\rho - ||\mathbf{c}^{(j)}||_{1}^\rho \right] p_i + \left[ ||\mathbf{c}^{(j)}||_{1}^\rho - (||\mathbf{c}^{(i)}||_{1}-c_i+c_j)^\rho \right] p_j \leq  \sum_{l=i+1}^{j-1}  \left[  (||\mathbf{c}^{(l)}||_{1}-c_i+c_j)^\rho - ||\mathbf{c}^{(l)}||_{1}^\rho \right] p_l. 
\end{eqnarray}
is satisfied, then $c_ip_j \leq c_jp_i$ may or may not hold. However, in the worst case the strategy satisfying $c_ip_j \leq c_jp_i$ may not be optimal for a given $\rho$. Hence, our argument in Eqn. \eqref{eqn20} is still valid since we are generating an upper bound for the optimal guessing strategy. However with the general condition the upper bound can be tightened at the expense of ending up with more complex expressions. For the asymptotic result of the paper, this simpler upper bound would be just sufficient. 

On the other hand, we notice that
\begin{align}
    \frac{P_X(x)}{\lceil c_x \rceil} = \frac{c_x P_X(x)}{\lceil c_x \rceil c_x} \geq \frac{P_X(x)}{c_x} \left( \frac{c_x}{\lceil c_x \rceil} \right)^{1+\rho}
\end{align}
from which the following inequality follows for $\rho \geq 0$,
\begin{align}
    \lceil c_x \rceil (P_X(x)/\lceil c_x \rceil)^{\frac{1}{1+\rho}} \geq c_x (P_X(x)/c_x)^\frac{1}{1+\rho}. \label{eqn311}
\end{align}

Thus, using the inequality \eqref{eqn311} and the pre-defined auxiliary random variable $Y$ earlier, we finally express the upper bound in a more compact form
\begin{align}
    \mathbb{E}[C_\mathcal{G^{*}}(X)^\rho] &\leq \left[\sum_x c_x (P_X(x)/c_x)^\frac{1}{1+\rho}\right]^{1+\rho} \\
    & \leq \left[ \sum_x \lceil c_x \rceil (P_X(x)/\lceil c_x \rceil)^{\frac{1}{1+\rho}} \right]^{1+\rho}  = \left[ \sum_y  P_Y(y)^{\frac{1}{1+\rho}} \right]^{1+\rho} \\\label{eqn711}
    &= \exp\{ \rho H_{\frac{1}{1+\rho}}(Y) \}
\end{align}

Notice that this upper bound will reduce to Arikan's upper bound i.e., $\exp(\rho H_{\frac{1}{1+\rho}}(X))$ with all costs set to unity. 

\section{Proof of Theorem \ref{thmexponent}}

Let us consider the general case and first define the induced random variables $Z_i \sim Z$  and  $Y_i \sim Y$ for the corresponding random variables $X_i \sim X$ for $i=\{1,\dots,n\}$ each with cost distributions $\mathcal{C}_i$ based on Definition \ref{def31}. Also let $\mathcal{F}^*$  and $\mathcal{H}^*$ be the induced optimal guessing strategies from $\mathcal{G}^*$ for random variables $\{Z_1, \dots, Z_n\}$ and $\{Y_1, \dots, Y_n\}$, respectively. 
 
 Now, consider the upper bound for i.i.d. random variables and observe
 {
\begin{eqnarray}
C_{\mathcal{G}^*}(x_1,\dots,x_n) 
&\leq& \sum_{\substack{
    x_1^\prime, x_2^\prime, \dots, x_n^\prime: \\
    C_{\mathcal{G}^*}(x_1^\prime,x_2^\prime,\dots,x_n^\prime) \leq C_{\mathcal{G}^*}(x_1,x_2,\dots,x_n) } } 
    \sum_{x_1^{''}}^{c_{x_1^\prime}} \dots \sum_{x_n^{\prime \prime}}^{c_{x_n^\prime}} \left(\prod_{i=1}^n \frac{c_{x_i}P_{X_i}(x_i^\prime)}{c_{x_i^\prime}P_{X_i}(x_i)} \right)^{\frac{1}{1+\rho}} \\
&\leq& \prod_{i=1}^n \sum_{x_i^\prime} 
c_{x_i^\prime}^{\frac{\rho}{1+\rho}}\left(\frac{c_{x_i}P_{X_i}(x_i^\prime)}{P_{X_i}(x_i)} \right)^{\frac{1}{1+\rho}} = \left[ \sum_{x_1^\prime} 
c_{x_1^\prime}^{\frac{\rho}{1+\rho}}\left(\frac{c_{x_1}P_{X_1}(x_1^\prime)}{P_{X_1}(x_1)} \right)^{\frac{1}{1+\rho}} \right]^n
\end{eqnarray}
}due to independence and series of inequalities $c_{x_1^\prime} P_X(x_1) \leq c_{x_1} P_X(x_1^{\prime}), c_{x_2^\prime} P_X(x_2) \leq c_{x_2} P_X(x_2^{\prime}), \dots, c_{x_n^\prime} P_X(x_n) \leq c_{x_n} P_X(x_n^{\prime})$ for all $\{x_i^{\prime}:C_{\mathcal{G}^*}(x_i^\prime) \leq C_{\mathcal{G}^*}(x_i)\}$ where $i=1,\dots,n$ that needs to hold for the optimal strategy $\mathcal{G}^*$ required by the necessary condition. Finally, we can upper bound the expected guessing cost for a sequence of i.i.d. random variables as
\begin{align}
\mathbb{E}[C_{\mathcal{G}^*}(X_1,\dots,X_n)^\rho] & = \sum_x P_{\boldsymbol{X}}(x_1,\dots,x_n) C_{\mathcal{G}^*}(x_1,\dots,x_n)^\rho \\
& \leq \prod_{i=1}^n \sum_{x_i}  P_{X_i}(x_i) \left[\sum_{x_i^\prime} 
c_{x_i^\prime}^{\frac{\rho}{1+\rho}}\left(\frac{c_{x_i}P_{X_i}(x_i^\prime)}{P_{X_i}(x_i)} \right)^{\frac{1}{1+\rho}}\right]^\rho =  \prod_{i=1}^n \left[\sum_{x_i} c_{x_i} (P_{X_i}(x_i)/c_{x_i})^\frac{1}{1+\rho}\right]^{(1+\rho)} 
\\
&  \leq \prod_{i=1}^n \left[ \sum_{y_i}  P_{Y_i}(y_i)^{\frac{1}{1+\rho}} \right]^{(1+\rho)} = \exp\left\{ \rho \sum_{i} H_{\frac{1}{1+\rho}}(Y_i) \right\} \label{eqn39}
\end{align}
where the last inequality follows due to inequalities similar to \eqref{eqn311} for each random variable $X_i$. If the cost and probability distributions of $X_i$'s are arranged such that the induced $Y_i$'s are identically distributed (for instance $X_i$'s are i.i.d. with identical cost distributions i.e., $\mathcal{C}_1 \equiv \mathcal{C}_2 \equiv \dots \equiv \mathcal{C}_n \triangleq \mathcal{C}$) then  we can further simplify \eqref{eqn39} as 
\begin{eqnarray}
\mathbb{E}[C_{\mathcal{G}^*}(X_1,\dots,X_n)^\rho] \leq \exp\{ \rho n H_{\frac{1}{1+\rho}}(Y_1) \}. 
\end{eqnarray}

Let us define $s_n = \mathbb{E}[C_{\mathcal{G}^*}(X_1,\dots,X_n)^\rho]$ and $\beta_k = \inf\{s_n: n \geq k\}$ for $k \geq 1$.  Note that $\beta_k$ is an increasing sequence ($\beta_{k+1} \geq \beta_k$) bounded above by \eqref{eqn39}. Then we have
\begin{align}
    \liminf_{n \rightarrow \infty} \frac{1}{n \rho} \ln(s_n) = \lim_{k \rightarrow \infty} \left\{ \frac{1}{n\rho} \ln(\beta_k) \right\} = \lim_{n \rightarrow \infty} \frac{1}{n} \sum_{i=1}^n H_{\frac{1}{1+\rho}}(Y_i) = \mathcal{R}_{\frac{1}{1+\rho}}(\{Y_i\}) \label{finupb}
\end{align}
which is defined to be the order-$1/(1+\rho)$ Rényi entropy rate \cite{Bunte2016} as long as the limit exists.  In addition to the upper bound, we can extend the lower bound given in \eqref{lowerboundalter} for a sequence of random variables as
\begin{align}
\mathbb{E}[C_{\mathcal{G}^*}(X_1,\dots,X_n)^\rho] \geq \mathbb{E}[C_{\mathcal{F}^*}(Z_1,\dots,Z_n)^\rho] \geq \left(1+ \ln \left( \prod_{i=1}^n \sum_{x_i} \lfloor c_{x_i} \rfloor \right)\right)^{-\rho} \exp \left\{ \rho \sum_{i=1}^n H_{\frac{1}{1+\rho}}(Z_i) \right\} \label{eqn47}
\end{align}
where the first inequality can be shown to be true through induction and the second inequality  follows from  \cite{arikan1996} through a bit of generalization. Note that $\mathcal{F}^*$ is the optimal induced strategy from $\mathcal{G}^*$. If the cost and probability distributions of $X_i$'s are arranged such that the induced $Z_i$'s are identically distributed (for instance $X_i$'s are i.i.d. with identical cost distributions i.e., $\mathcal{C}_1 \equiv \mathcal{C}_2 \equiv \dots \equiv \mathcal{C}_n \triangleq \mathcal{C}$) then  we can further simplify \eqref{eqn47} as 
{
\begin{eqnarray}
\mathbb{E}[C_{\mathcal{G}^*}(X_1,\dots,X_n)^\rho] \geq \left(1+ \ln \left( \prod_{i=1}^n \sum_{x_i} \lfloor c_{x_i} \rfloor \right)\right)^{-\rho} \exp \left\{ \rho n H_{\frac{1}{1+\rho}}(Z_i) \right\}. \label{eqlb120}
\end{eqnarray}}

Similarly, we further define $\alpha_k = \sup\{s_n: n \geq k\}$ for $k \geq 1$ which makes it a decreasing sequence lower bounded by \eqref{eqlb120}. As a consequence, we have
{
\begin{align}
    \limsup_{n \rightarrow \infty} \frac{1}{n\rho} \ln(s_n) = \lim_{k \rightarrow \infty} \left\{ \frac{1}{n\rho} \ln(\alpha_k) \right\} &= \lim_{n \rightarrow \infty} \ln \left(1 +  \ln \left(\prod_{i=1}^n \sum_{x_i} \lfloor c_{x_i} \rfloor \right) \right)^{-1/n} +  \lim_{n \rightarrow \infty} \frac{1}{n} \sum_{i=1}^n H_{\frac{1}{1+\rho}}(Z_i) \nonumber \\
    &= \lim_{n \rightarrow \infty} \frac{1}{n} \sum_{i=1}^n H_{\frac{1}{1+\rho}}(Z_i) = \mathcal{R}_{\frac{1}{1+\rho}}(\{Z_i\}). \label{finaleqn}
\end{align}}

If $\{X_i\}$ are indentically distributed with the same cost distribution $\mathcal{C}$, then  order-$1/(1+\rho)$ Rényi entropy rates would be equal to $H_{\frac{1}{1+\rho}}(Y)$ and $H_{\frac{1}{1+\rho}}(Z)$, respectively. Note that in general, these rates are not necessarily equal. However, if the costs are integers, {it would not be hard to verify} $\mathcal{R}_{\frac{1}{1+\rho}}(\{Y_i\}) = \mathcal{R}_{\frac{1}{1+\rho}}(\{Z_i\})$. Thus, combining equations \eqref{finaleqn} with \eqref{finupb}, we shall have 
\begin{align}
    \limsup_{n \rightarrow \infty} \frac{1}{n \rho} \ln(s_n) = \liminf_{n \rightarrow \infty} \frac{1}{n \rho} \ln(s_n)     &=\mathcal{R}_{\frac{1}{1+\rho}}(\{Y_i\}) =\mathcal{R}_{\frac{1}{1+\rho}}(\{Z_i\}) \\
    &=\mathcal{R}_{\frac{1}{1+\rho}}(\{X_i\})
\end{align}
which completes the proof. 

\section{Proof of Theorem \ref{sasonbound1}}

Let us first consider $\rho \geq 1$, and for a given real constant $c \geq 1$ we let $r(u;c)$ be the parametric function given by
\begin{align}
    r(u;c) = \frac{1}{c(1+\rho)} \left( u^{1+\rho} - (u-c)^{1+\rho} - c^{1+\rho} \right) - (u-c)^{\rho}, \ \ \ \ u \geq c
\end{align}
One of the things we realize about this function is that its derivative is non-negative, i.e., 
\begin{eqnarray}
\frac{\partial }{\partial u} r(u;c) = \frac{1}{c} (u^\rho - (u-c)^\rho) - \rho(u-c)^{\rho-1} \geq 0 
\end{eqnarray}
which is not hard to see by invoking mean value theorem from standard calculus. Moreover, we have 
\begin{align}
    \frac{\partial }{\partial c}  \frac{\partial }{\partial u} r(u;c) = \frac{\rho}{c} (u-c)^{\rho-1} - \frac{1}{c^2} (u^\rho - (u-c)^\rho)+ \rho(\rho-1)(u-c)^{\rho-2} \geq 0,
\end{align}
i.e., it is always non-negative for $u \geq c \geq 1$ and $\rho \geq 1$. Thus $\frac{\partial }{\partial u} r(u,c)$ is an increasing function of $c$ which therefore leads to the conclusion that for $c \geq 1$, it is non-negative. Since $r(c,c)=0$, it follows that $r(u;c)$ is non-negative for $u\geq c \geq 1$. 

Remember that for a given random variable $X$ associated with costs $\mathcal{C}={c_1,\dots,c_M}$, we have from equation \eqref{eqn63} 
\begin{align}
    \mathbb{E}[C_\mathcal{G}(X)^\rho] -  \mathbb{E}[(C_\mathcal{G}(X)-\overline{c}_X(\rho))^\rho] \leq \left[ \sum_{k=1}^{M^\prime} q_k^{1/\rho}\right]^\rho, \label{eqn111}
\end{align}
where $M^\prime = \sum_{i=1}^M \lceil c_i \rceil$,  $q_k = p_i \textrm{ for } \sum_{l=1}^{i-1} \lceil c_l \rceil < k \leq  \sum_{l=1}^{i} \lceil c_l \rceil \textrm{ and } i = 1,\dots,M$ and balancing cost  $\overline{c}_X(\rho)$ is as defined in Definition \ref{defbalcost} for as long as  $q_k$'s satisfy the relation given in equation \eqref{eqn64}. Note that since $\mathbb{E}[r(C_\mathcal{G}(X), \overline{c}_X(\rho) )] \geq 0$, it implies for $\rho \geq 1$ that if $\overline{c}_X(\rho) \leq \overline{c}_X(1+\rho)$, 
\begin{align}
\mathbb{E}[(C_\mathcal{G}(X)-\overline{c}_X(\rho))^\rho] &\leq \frac{1}{\overline{c}_X(\rho)(1+\rho)} \left(\mathbb{E}[C_\mathcal{G}(X)^{1+\rho}] -  \mathbb{E}[(C_\mathcal{G}(X)-\overline{c}_X(\rho))^{1+\rho}]\right)- \frac{\overline{c}_X^{\rho}(\rho)}{1+\rho} \\
& \leq  \frac{1}{\overline{c}_X(\rho)(1+\rho)} \left(\mathbb{E}[C_\mathcal{G}(X)^{1+\rho}] -  \mathbb{E}[(C_\mathcal{G}(X)-\overline{c}_X(1+\rho))^{1+\rho}]\right)- \frac{\overline{c}_X^{\rho}(\rho)}{1+\rho} \\
&\leq \frac{1}{\overline{c}_X(\rho)(1+\rho)} \left[ \sum_{k=1}^{M^\prime} q_k^{\frac{1}{1+\rho}}\right]^{1+\rho} - \frac{\overline{c}_X^{\rho}(\rho)}{1+\rho}. \label{eqn118}
\end{align}
where inequality \eqref{eqn118} follows for as long as 
$q_{k+1}^{\frac{1}{1+\rho}} \leq \frac{1}{k} (q_1^{\frac{1}{1+\rho}}+\dots+q_{k}^{\frac{1}{1+\rho}}), \textrm{  for } k=1,\dots, M^\prime-1$. However if equation \eqref{eqn64} is already satisfied for $\rho \geq 1$, then this inequality would also be satisfied. Hence combining it with the equation \eqref{eqn111}, we finally obtain
\begin{eqnarray}
\mathbb{E}[C_\mathcal{G}(X)^\rho] \leq \frac{1}{\overline{c}_X(\rho)(1+\rho)} \left[ \sum_{k=1}^{M^\prime} q_k^{\frac{1}{1+\rho}}\right]^{1+\rho} + \left[ \sum_{k=1}^{M^\prime} q_k^{1/\rho}\right]^\rho - \frac{\overline{c}_X^{\rho}(\rho)}{1+\rho} . \label{eqn115}
\end{eqnarray}

On the other hand, if $\overline{c}_X(1+\rho) \leq \overline{c}_X(\rho)$, then we use the fact that $\mathbb{E}[r(C_\mathcal{G}(X)), \overline{c}_X(1+\rho) ] \geq 0$ which implies for $\rho \geq 1$,
\begin{align}
    \mathbb{E}[(C_\mathcal{G}(X)-\overline{c}_X(1+\rho))^\rho] &\leq \frac{1}{\overline{c}_X(1+\rho)(1+\rho)} \left(\mathbb{E}[C_\mathcal{G}(X)^{1+\rho}] -  \mathbb{E}[(C_\mathcal{G}(X)-\overline{c}_X(1+\rho))^{1+\rho}]\right)- \frac{\overline{c}_X^{\rho}(1+\rho)}{1+\rho} \\
&\leq \frac{1}{\overline{c}_X(1+\rho)(1+\rho)} \left[ \sum_{k=1}^{M^\prime} q_k^{\frac{1}{1+\rho}}\right]^{1+\rho} - \frac{\overline{c}_X^{\rho}(1+\rho)}{1+\rho} \label{eqn117}
\end{align}
Hence using equation \eqref{eqn111} and \eqref{eqn117} , we finally obtain
\begin{align}
        \mathbb{E}[C_\mathcal{G}(X)^\rho]  &\leq \left[ \sum_{k=1}^{M^\prime} q_k^{1/\rho}\right]^\rho + \mathbb{E}[(C_\mathcal{G}(X)-\overline{c}_X(\rho))^\rho] \\
        &\leq \left[ \sum_{k=1}^{M^\prime} q_k^{1/\rho}\right]^\rho + \mathbb{E}[(C_\mathcal{G}(X)-\overline{c}_X(1+\rho))^\rho] \\
        &\leq \frac{1}{\overline{c}_X(1+\rho)(1+\rho)} \left[ \sum_{k=1}^{M^\prime} q_k^{\frac{1}{1+\rho}}\right]^{1+\rho} + \left[ \sum_{k=1}^{M^\prime} q_k^{1/\rho}\right]^\rho - \frac{\overline{c}_X^{\rho}(1+\rho)}{1+\rho} \label{eqn120}
\end{align}
Thus combining \eqref{eqn115} and \eqref{eqn120}, we get the final expression 
\begin{align}
    \mathbb{E}[C_\mathcal{G}(X)^\rho] \leq \frac{1}{\min\{\overline{c}_X(\rho),\overline{c}_X(1+\rho)\}(1+\rho)} \left[ \sum_{k=1}^{M^\prime} q_k^{\frac{1}{1+\rho}}\right]^{1+\rho} + \left[ \sum_{k=1}^{M^\prime} q_k^{1/\rho}\right]^\rho - \frac{[\min\{\overline{c}_X(\rho),\overline{c}_X(1+\rho)\}]^\rho}{1+\rho} \label{eqn123}
\end{align}

Let us now consider the case $\rho \in (0,1)$. We introduce the following function this time,
\begin{align}
    r(u,c) = \frac{1}{c(1+\rho)}(u^{1+\rho} - (u-c)^{1+\rho} + \rho c^{1+\rho}) - u^{\rho}, \ \ \ \ u \geq c \geq 1. 
\end{align}
When we take the derivative with respect to $u$, we get
\begin{align}
    \frac{\partial }{\partial u} r(u,c) &= \frac{1}{c} (u^\rho - (u-c)^\rho) - \rho u^{\rho-1} \\
    & = \rho x^{\rho-1} -  \rho u^{\rho-1}, \ \ x\in(u-c,u)  \label{MVT}\\
    & > 0,
\end{align}
where  the equation \eqref{MVT} holds due to mean value theorem for all $\rho \in (0,1)$. Since $r(c,c) = 0$, we always have $r(u,c) \geq 0$ for all $u \geq c \geq 1$. Applying this function, we will have $\mathbb{E}[r(C_\mathcal{G}(X)), \overline{c}_X(1+\rho) ] \geq 0$ meaning that
\begin{align}
    \mathbb{E}[C_\mathcal{G}(X)^\rho] &\leq \frac{1}{\overline{c}_X(1+\rho)(1+\rho)} \left(\mathbb{E}[C_\mathcal{G}(X)^{1+\rho}] -  \mathbb{E}[(C_\mathcal{G}(X)-\overline{c}_X(1+\rho))^{1+\rho}] + \rho \overline{c}_X^{1+\rho}(1+\rho) \right) \\
    &\leq \frac{1}{\overline{c}_X(1+\rho)(1+\rho)}  \left[ \sum_{k=1}^{M^\prime} q_k^{\frac{1}{1+\rho}}\right]^{1+\rho} + \frac{\rho}{1+\rho} \overline{c}_X^{\rho}(1+\rho) \label{eqn131}
\end{align}
where the inequality \eqref{eqn131} follows since $1 + \rho \geq 1$. Similarly for $\mathbb{E}[r(C_\mathcal{G}(X)), \overline{c}_X(\rho) ] \geq 0$, we get
\begin{align}
    \mathbb{E}[C_\mathcal{G}(X)^\rho] &\leq \frac{1}{\overline{c}_X(\rho)(1+\rho)} \left(\mathbb{E}[C_\mathcal{G}(X)^{1+\rho}] -  \mathbb{E}[(C_\mathcal{G}(X)-\overline{c}_X(\rho))^{1+\rho}] + \rho \overline{c}_X^{1+\rho}(\rho) \right) \\
    &\leq \frac{1}{\overline{c}_X(\rho)(1+\rho)}  \left[ \sum_{k=1}^{M^\prime} q_k^{\frac{1}{1+\rho}}\right]^{1+\rho} + \frac{\rho}{1+\rho} \overline{c}_X^{\rho}(\rho)
\end{align}
if $ \overline{c}_X(\rho) \leq  \overline{c}_X(1+\rho)$. Thus,
\begin{align}
\mathbb{E}[C_\mathcal{G}(X)^\rho] \leq \frac{1}{\min\{\overline{c}_X(\rho),\overline{c}_X(1+\rho)\} (1+\rho)} \left[ \sum_{k=1}^{M^\prime} q_k^{\frac{1}{1+\rho}}\right]^{1+\rho} + \frac{\rho}{1+\rho} [\min\{\overline{c}_X(\rho),\overline{c}_X(1+\rho)\}]^\rho \label{eqn132}
\end{align}
Thus, using an indicator function $\mathbbm{1}_{\rho \geq 1}$ to be able to combine \eqref{eqn123} and \eqref{eqn132}, the result follows.





\section{Proof of Theorem \ref{Thm36}}

Let us consider the case $\rho \in (0,1)$ case first and state the following Lemma. 

\begin{lemma} \label{lemma178}
For $\rho \in (0,1)$ and any $c \in  \mathbb{R}, c \geq 1$ and $u \geq c$,
\begin{align}
    u^\rho \leq \frac{u^{1+\rho}-(u-c)^{1+\rho}}{c(1+\rho)} + \frac{\rho c^\rho}{1+\rho} \mathbbm{1}_{\{c\leq u < c+1\}} + \left((c+1)^\rho - \frac{(c+1)^{1+\rho}-1}{c(1+\rho)}\right)\mathbbm{1}_{\{u \geq c+1\}}
\end{align}
\end{lemma}

\begin{proof}
For $\rho \in (0,1)$ and a given real constant $c \geq 1$ we define to parametric functions given by
\begin{align}
    r_1(u,c) &= \frac{u^{1+\rho}-(u-c)^{1+\rho}}{c(1+\rho)} + \frac{\rho c^\rho}{1+\rho} - u^\rho, \\
    r_2(u,c) &= \frac{u^{1+\rho}-(u-c)^{1+\rho}}{c(1+\rho)} + (c+1)^\rho - \frac{(c+1)^\rho-1}{c(1+\rho)} - u^\rho 
\end{align}
For $u \in (c,\infty)$, we have $\frac{\partial }{\partial u}  r_1(u,c) = \frac{\partial }{\partial u} r_2(u,c) = \frac{1}{c}(u^\rho - (u-c)^\rho) -\rho u^{\rho-1} > 0$ again similarly due to mean value theorem. Moreover, $r_1(c,c) = r_2(c+1,c) = 0$. As a result, $r_1(u,c) \geq 0$ for $u \geq c$ and $r_2(u,c) \geq 0$ for $u \geq c+1$. Next we observe that for $\rho \in (0,1)$ and $c \geq 0$, we have $(c\rho-1)(c+1)^\rho < \rho c^{1+\rho}-1$ which implies that
\begin{align}
(c+1)^\rho - \frac{(c+1)^\rho-1}{c(1+\rho)} < \frac{\rho c^\rho}{1+\rho}
\end{align}
which completes the proof by recognizing $\min\{r_1(u,c),r_2(u,c)\} = r_2(u,c)$. 
\end{proof}

Let $\overline{c}_{{min}_X}(\rho) = \min\{\overline{c}_X(\rho),\overline{c}_X(1+\rho)\}$. Now using the result of Lemma \ref{lemma178} for $\rho \in (0,1)$ and replacing $u$ with $C_\mathcal{G}(X)$, and considering both cases $\overline{c}_X(1+\rho) \leq \overline{c}_X(\rho)$, $\overline{c}_X(\rho) \leq \overline{c}_X(1+\rho)$ separately, similar to Appendix E, we obtain
\begin{align}
    \mathbb{E}[C_\mathcal{G}(X)^\rho] &\leq \frac{1}{ \overline{c}_{{min}_X}(\rho) (1+\rho)} \left(\mathbb{E}[C_\mathcal{G}(X)^{1+\rho}] -  \mathbb{E}[(C_\mathcal{G}(X)-\overline{c}_X(1+\rho))^{1+\rho}]  \right) \nonumber \\ & \ \ \  +  \frac{\rho \overline{c}^\rho_{{min}_X}(\rho)}{1+\rho} P(\overline{c}_{{min}_X}(\rho) \leq C_\mathcal{G}(X) < \overline{c}_{{min}_X}(\rho)+1) \nonumber \\ 
    & \ \ \ + \left((\overline{c}_{{min}_X}(\rho)+1)^\rho - \frac{(\overline{c}_{{min}_X}(\rho)+1)^{1+\rho}-1}{\overline{c}_{{min}_X}(\rho)(1+\rho)}\right) P(C_\mathcal{G}(X) \geq \overline{c}_{{min}_X}(\rho)+1)
    \\
    & \leq \frac{1}{\overline{c}_{{min}_X}(\rho) (1+\rho)}  \left[ \sum_{k=1}^{M^\prime} q_k^{\frac{1}{1+\rho}}\right]^{1+\rho} +  \frac{\rho \overline{c}^\rho_{{min}_X}(\rho)}{1+\rho} P(\overline{c}_{{min}_X}(\rho) \leq C_\mathcal{G}(X) < \overline{c}_{{min}_X}(\rho)+1) \nonumber \\ 
    & \ \ \ + \left((\overline{c}_{{min}_X}(\rho)+1)^\rho - \frac{(\overline{c}_{{min}_X}(\rho)+1)^{1+\rho}-1}{\overline{c}_{{min}_X}(\rho)(1+\rho)}\right) P(C_\mathcal{G}(X) \geq \overline{c}_{{min}_X}(\rho)+1) \label{eqn138}
\end{align}
where equation \eqref{eqn138} follows from \eqref{eqn111}. Next, we state the following Lemma. 
\begin{lemma}
For $\rho \in [1,2]$ and any $c\in \mathbb{R}, c\geq 1$ and $u \geq c$,
\begin{align}
    u^\rho \leq \frac{u^{1+\rho}-(u-c)^{1+\rho}}{1+\rho} + \frac{u^{\rho}-(u-c)^{\rho}}{\rho} + \frac{c^\rho(\rho^2-c\rho-1)}{\rho(1+\rho)} \label{eqn139}
\end{align}
\end{lemma}
\begin{proof}
For $\rho \in [1,2]$, let $r(u,c)$ be a parametric function given by
\begin{align}
    r(u,c) = \frac{u^{1+\rho}-(u-c)^{1+\rho}}{1+\rho} + \frac{u^{\rho}-(u-c)^{\rho}}{\rho} - u^\rho + \frac{c^\rho(\rho^2-c\rho-1)}{\rho(1+\rho)}, \ \ \ u \geq c
\end{align}
If we take the partial derivative of this function with respect to $u$, we get
\begin{align}
\frac{\partial }{\partial u} r(u,c) &= u^\rho-(u-c)^\rho + u^{\rho-1}-(u-c)^{\rho-1} - \rho u^{\rho-1} \\
&\geq c + \rho (u-c)^{\rho-1} + u^{\rho-1}-(u-c)^{\rho-1} - \rho u^{\rho-1} \label{convexity1} \\
&= c + (\rho-1)((u-c)^{\rho-1}-u^{\rho-1}) \geq 2c - \rho c \geq 0 \label{eqn143}
\end{align}
where equation $\eqref{convexity1}$ follows from the convexity of $f(x) = x^\rho$ function in $(c,\infty)$ for $\rho \geq 1$ and equation $\eqref{eqn143}$ holds due to 
\begin{align}
    -c \leq (u-c)^{\rho-1} - u^{\rho-1} \leq 0
\end{align}
for $\rho \in [1,2]$ and $u \geq c$. Finally, note that $r(c,c) = 0$ implying that $r(u,c) \geq 0$ for $\rho \in [1,2]$ and $u \geq c$. 
\end{proof}

Replacing $u$ in \eqref{eqn139} with $C_\mathcal{G}(x)$, $c$ with the balancing cost, and taking the expectation of both sides, and considering both cases $\overline{c}_X(1+\rho) \leq \overline{c}_X(\rho)$, $\overline{c}_X(\rho) \leq \overline{c}_X(1+\rho)$ separately, similar to Appendix E, we finally obtain
\begin{align}
    \mathbb{E}[C_\mathcal{G}(X)^\rho] & \leq \frac{1}{1+\rho}   \left(\mathbb{E}[C_\mathcal{G}(X)^{1+\rho}] -  \mathbb{E}[(C_\mathcal{G}(X)-\overline{c}_X(1+\rho))^{1+\rho}]  \right) \\
    & \ \ \ + \frac{1}{\rho} \left(\mathbb{E}[C_\mathcal{G}(X)^{\rho}] -  \mathbb{E}[(C_\mathcal{G}(X)-\overline{c}_X(\rho))^{\rho}]  \right) + \frac{\overline{c}_{{min}_X}^\rho(\rho)(\rho^2-\overline{c}_{{min}_X}(\rho)\rho-1)}{\rho(1+\rho)}
\end{align}
from which the result follows using the relationship given in \eqref{eqn63}. 

\section{Proof of Theorem \ref{Thm37}}

Considering $\rho \geq 2, u \geq c > 1$ for some $c \in \mathbb{R}$, we define the parametric function, 
\begin{align}
    r(u,c) =  \frac{u^{1+\rho}-(u-c)^{1+\rho}}{c(1+\rho)} - u^{\rho} + \frac{\rho c u^{\rho-1}}{2} - \frac{\rho(\rho-1)}{2(1+\rho)}
\end{align}
One of the things we realize about this function is that its derivative is non-negative for $\rho \geq 2$, i.e.,
\begin{align}
\frac{\partial }{\partial u} r(u,c) = \frac{1}{c} (u^\rho - (u-c)^\rho) - \rho u^{\rho-1} + \frac{1}{2}\rho (\rho-1) c u^{\rho-2} \label{eqn147}
\end{align}
To be able to shorten the notation let $v(x) = x^\rho$. Now consider the Taylor series expansion of $v(x)$ around $u$ i.e.,
\begin{align}
    v(x) = v(u) + v'(u)(x-u) + \frac{1}{2}v''(u)(x-u)^2 + \frac{1}{6}v'''(u)(x-u)^3 + \dots 
\end{align}
Now, evaluate $v(x)$ at $x=u-c$ and observe to truncate the expansion, 
\begin{align}
    (u-c)^\rho = v(u-c) = v(u) - c v'(u) + \frac{c^2}{2} v''(u) - \frac{c^3}{6} v'''(u') \label{eqn149}
\end{align}
for some $u' \in (u-c,u)$ for $u \geq c$. By plugging \eqref{eqn149} into \eqref{eqn147}, we obtain
\begin{align}
    \frac{\partial }{\partial u} r(u,c) &= v'(u) - \frac{c}{2}v''(u) +\frac{c^2}{6}v'''(u') - \rho u^{\rho-1} + \frac{1}{2}\rho (\rho-1) c u^{\rho-2} \\
    & = \frac{c^2}{6}v'''(u') = \frac{c^2}{6} \rho (\rho-1) (\rho-2) ({u'})^{\rho-3} \geq 0
\end{align}
since $u' > 0$ due to $u\geq c$ and $\rho \geq 2$. It is also not hard to verify that $r(c,c) = 0$, which eventually implies that $r(u,c) \geq 0$ for $u \geq c$. If we substitute $u=C_\mathcal{G}(X)$ and $c=\overline{c}_X(1+\rho)$ and take the expectation i.e. $\mathbb{E}[r(u,c)]\geq 0$, we finally obtain
\begin{align}
\mathbb{E}[C_\mathcal{G}(X)^\rho] & \leq \frac{1}{\overline{c}_X(1+\rho)(1+\rho)} \left(\mathbb{E}[C_\mathcal{G}(X)^{1+\rho}] -  \mathbb{E}[(C_\mathcal{G}(X)-\overline{c}_X(1+\rho))^{1+\rho}]  \right)   + \frac{\rho \overline{c}_X(1+\rho)}{2} \mathbb{E}[C_\mathcal{G}(X)^{\rho-1}] - \frac{\rho(\rho-1)}{2(1+\rho)}  \label{eqn153}
\end{align}
from which the result follows using the relationship given in \eqref{eqn63}.

%









\begin{thebibliography}{1}


\bibitem{pliam99} J. O. Pliam. (1999) The Disparity Between Work and Entropy in Cryptology. Available Online: http://philby.ucsd.edu/cryptolib/1998/98-24.html


\bibitem{massey1}  J. L. Massey, “Guessing and entropy,” \textit{in Proc. IEEE Int. Symp. on Information Theory} Trondheim, Norway, 1994, pp. 204.

\bibitem{mceliceyu} R. J. McEliece and Z. Yu, “An inequality on entropy,” \textit{in Proc. IEEE International Symposium on Information Theory}, p. 329,
Whistler, Canada, September 1995.

\bibitem{arikan1996} E. Arikan, “An inequality on guessing and its application to sequential
decoding,” \textit{IEEE Trans. Inform. Theory,} vol. 42, pp. 99–105, Jan. 1996.


\bibitem{Bozdas1997} S. Boztas, “Comments on "An inequality on guessing and its application to sequential decoding"," in \textit{IEEE Transactions on Information Theory,} vol. 43, no. 6, pp. 2062-2063, Nov. 1997.


\bibitem{sason2018} I. Sason and S. Verdú, “Improved Bounds on Guessing Moments via Rényi Measures," \textit{in Proc. IEEE International Symposium on Information Theory} (ISIT), Vail, CO, 2018, pp. 566-570.

\bibitem{cambell56} L. L. Campbell, “A coding theorem and Rényi's entropy." \textit{Information and control} 8.4: 423-429, 1965.

\bibitem{kuzuoka2019} S. Kuzuoka, “On the Conditional Smooth Rényi Entropy and Its Application in Guessing," \textit{in Proc. IEEE International Symposium on Information Theory (ISIT),} Paris, France, 2019, pp. 647-651.

\bibitem{duffy2019} K. R. Duffy, J. Li and M. Médard, “Capacity-Achieving Guessing Random Additive Noise Decoding," in \textit{IEEE Transactions on Information Theory,} vol. 65, no. 7, pp. 4023-4040, July 2019.

\bibitem{bracher2019} A. Bracher, E. Hof and A. Lapidoth, ``Guessing Attacks on Distributed-Storage Systems," in \textit{IEEE Transactions on Information Theory}, vol. 65, no. 11, pp. 6975-6998, Nov. 2019.

\bibitem{Graczyk2022} R. Graczyk, A. Lapidoth, N. Merhav and C. Pfister, "Guessing Based on Compressed Side Information," in \textit{IEEE Transactions on Information Theory}, vol. 68, no. 7, pp. 4244-4256, July 2022.

\bibitem{malone2004} D. Malone and W. G. Sullivan, “Guesswork and entropy," in \textit{IEEE Transactions on Information Theory,} vol. 50, no. 3, pp. 525-526, March 2004.

\bibitem{pfister2004} C. E. Pfister and W. G. Sullivan, "Renyi entropy, guesswork moments, and large deviations," in IEEE Transactions on Information Theory, vol. 50, no. 11, pp. 2794-2800, Nov. 2004,

\bibitem{arslanisit2020} S. S. Arslan and E. Haytaoglu, "Cost of Guessing: Applications to Data Repair," \textit{in Proc. IEEE International Symposium on Information Theory} (ISIT), 2020, pp. 2194-2198. 

\bibitem{arslanisit2022} S. S. Arslan and E. Haytaoglu, "Improved Bounds on the Moments of Guessing Cost," \textit{in Proc. IEEE International Symposium on Information Theory} (ISIT), Espoo, Finland, 2022, pp. 3351-3356.

\bibitem{noroozi2019} M. Noroozi and Z. Eslami, ``Public-key encryption with keyword search:
A generic construction secure against online and offline keyword guessing attacks,’’ J. Ambient Intell. Humanized Comput., 2019.

\bibitem{Song2000} D. X. Song, D. Wagner and A. Perrig, "Practical techniques for searches on encrypted data," \textit{in Proc. IEEE Symposium on Security and Privacy,} S\&P 2000, 2000, pp. 44-55.


\bibitem{Ulam1988} Ivan Niven, ``Coding Theory Applied to a Problem of Ulam.'' Math. Mag. 61 (1988) 275-281. 


\bibitem{ArikanGwLies} E. Arikan and S. Boztas, ``Guessing with lies,'' in IEEE International Symposium on Information Theory, Lausanne, Switzerland, 2002, pp. 208-,. 

\bibitem{duffy2022} K. R. Duffy, M. Médard and W. An, ``Guessing Random Additive Noise Decoding With Symbol Reliability Information (SRGRAND)," in `\textit{IEEE Transactions on Communications,} vol. 70, no. 1, pp. 3-18, Jan. 2022.





\bibitem{solomon2020} A. Solomon, K. R. Duffy and M. Médard, "Soft Maximum Likelihood Decoding using GRAND," \textit{in Proc. IEEE International Conference on Communications} (ICC), 2020, pp. 1-6. 


\bibitem{haytaoglu2022} E. Haytaoglu, E. Kaya and S. S. Arslan, ``Data Repair-Efficient Fault Tolerance for Cellular Networks Using LDPC Codes," in \textit{IEEE Transactions on Communications,} vol. 70, no. 1, pp. 19-31, Jan. 2022.

\bibitem{renyi1960} A. Renyi, “On Measures of Information and Entropy,” Proc. 4th Berkeley Symp. Math.,
Stat. Prob., vol. 1, pp. 547-561, 1960.

\bibitem{LDPC1963} R. Gallager, “Low-density parity-check codes," in \textit{IRE Transactions on Information Theory,} vol. 8, no. 1, pp. 21-28, January 1962.


\bibitem{Bunte2016} C. Bunte and A. Lapidoth, "Maximum Rényi Entropy Rate," in \textit{IEEE Transactions on Information Theory}, vol. 62, no. 3, pp. 1193-1205, March 2016.


\bibitem{lubyandortree1998} M. Luby, M. Mitzenmacher, and A. Shokrollahi, ``Analysis of random
processes via and-or tree evaluation,” \textit{in Proc. 9th Annu. ACM-SIAM
Symp. Discrete Algorithms,} 1998, pp. 364–373.

\bibitem{Le2017} M. Le, Z. Song, Y. Kwon and E. Tilevich, ``Reliable and efficient mobile edge computing in highly dynamic and volatile environments," \textit{Second International Conference on Fog and Mobile Edge Computing} (FMEC), 2017, pp. 113-120.

\bibitem{Amraoui2007} A. Amraoui, A. Montanari, and R. Urbanke, ``How to find good finitelength codes: From art towards science,” in \textit{Europ. Trans. Telecomm.,} vol. 18, pp. 491–508, 2007.

\bibitem{Richardson2001}  T. J. Richardson, M. A. Shokrollahi and R. L. Urbanke, ``Design of capacity-approaching irregular low-density parity-check codes," in \textit{IEEE Transactions on Information Theory}, vol. 47, no. 2, pp. 619-637, Feb 2001. 

\bibitem{lai2016} W. K. Lai and   E. Kim, ``Some inequalities involving geometric and harmonic means," \textit{In International Mathematical Forum,} (2016), Vol. 11, No. 4, pp. 163-169.

\bibitem{park2017ldpc} H. Park, D. Lee and J.  Moon, ``LDPC code design for distributed storage: Balancing repair bandwidth, reliability, and storage overhead," in \textit{IEEE Transactions on Communications,} 66(2), 507-520, 2017.

\bibitem{Dragomir98} S. S. Dragomir and J. v. d. Hoek. ``New inequalities for the moments of guessing mapping," \textit{East Asian Math. Journal}, 14 (1)(1998), pp. 1-14. 

\bibitem{Williams1964algorithm}
J. Williams, ``Algorithm 232: Heapsort," \textit{Communications of the ACM}, 7 (6) (1964), pp. 347–348.

\bibitem{mergeSort}
D. Knuth . ``Section 5.2.4: Sorting by Merging". Sorting and Searching. The Art of Computer Programming." 3 (2nd ed.) (1998), Addison-Wesley. pp. 158–168.
\bibitem{Brocher2015} A. Bracher, E. Hof and A. Lapidoth, ``Guessing Attacks on Distributed-Storage Systems," 2015 IEEE International Symposium on Information Theory (ISIT), Hong Kong, China, 2015, pp. 1585-1589, doi: 10.1109/ISIT.2015.7282723.

\bibitem{Christiansen2015} M. M. Christiansen, K. R. Duffy, F. du Pin Calmon and M. Médard, "Multi-User Guesswork and Brute Force Security,`` in IEEE Transactions on Information Theory, vol. 61, no. 12, pp. 6876-6886, Dec. 2015, doi:
\bibitem{Huleihel2017} W. Huleihel, S. Salamatian and M. Médard, ``Guessing with limited memory," 2017 IEEE International Symposium on Information Theory (ISIT), Aachen, Germany, 2017, pp. 2253-2257, doi: 10.1109/ISIT.2017.8006930.



\bibitem{Kumar2022} M. A. Kumar, A. Sunny, A. Thakre, A. Kumar and G. D. Manohar, ``A Unified Framework for Problems on Guessing, Source Coding, and Tasks Partitioning," 2022 IEEE International Symposium on Information Theory (ISIT), Espoo, Finland, 2022, pp. 3339-3344, doi: 10.1109/ISIT50566.2022.9834851.






\end{thebibliography}
\end{document}